\documentclass[journal,two column]{IEEEtran}

\usepackage{slashbox}
\usepackage{amsmath}
\usepackage{blkarray}
\usepackage{multirow}
\usepackage{amssymb}
\usepackage{mathtools}
\usepackage{amsthm}

\newtheorem{theorem}{Theorem}

\newtheorem{lemma}[theorem]{Lemma}
\newtheorem{corollary}[theorem]{Corollary}
\newtheorem{algorithm}[theorem]{Algorithm}

\DeclareMathOperator*{\argmax}{arg\,max}
\DeclareMathOperator*{\argmin}{arg\,min}

\title{\huge{An LP approach for Solving Two-Player Zero-Sum Repeated Bayesian Games}}

\author{Lichun Li, Cedric Langbort and Jeff Shamma}

\begin{document}

\maketitle

\begin{abstract}
This paper studies two-player zero-sum repeated Bayesian games in which every player has a private type that is unknown to the other player, and the initial probability of the type of every player is publicly known. The types of players are independently chosen according to the initial probabilities, and are kept the same all through the game. At every stage, players simultaneously choose actions, and announce their actions publicly. For finite horizon cases, an explicit linear program is provided to compute players' security strategies. Moreover, based on the existing results in \cite{rosenberg1998duality}, this paper shows that a player's sufficient statistics, which is independent of the strategy of the other player, consists of the belief over the player's own type, the regret with respect to the other player's type, and the stage. Explicit linear programs are provided to compute the initial regrets, and the security strategies that only depends on the sufficient statistics. For discounted cases, following the same idea in the finite horizon, this paper shows that a player's sufficient statistics consists of the belief of the player's own type and the anti-discounted regret with respect to the other player's type. Besides, an approximated security strategy depending on the sufficient statistics is provided, and an explicit linear program to compute the approximated security strategy is given. This paper also obtains a bound on the performance difference between the approximated security strategy and the security strategy.
\end{abstract}

\section{Introduction}
In many strategic decision-making situations, players do not have complete information about the other decision-makers' characteristics or payoffs. This kind of situations is typically modeled as as so-called Bayesian game, in which each player has a type whose realization is privately observed, although a prior distribution is common knowledge.

In most prior work in the literature of Bayesian games, one player's belief on the types of the other players plays an important role in figuring out the Nash equilibrium or the perfect Bayesian equilibrium. Generally speaking, a player's belief of the other players' type depends on the other players' strategies. A special class of Bayesian games in which the common information based beliefs are strategy independent was considered in \cite{nayyar2014common}. Because of the decoupling between the strategies and the beliefs, a backward induction algorithm was given to find Nash equilibria of the game. The cases when the beliefs of the players are strategy dependent were considered in \cite{ouyang2017dynamic,vasal2016systematic}. Both papers used perfect Bayesian equilibrium as their solution concept. Perfect Bayesian equilibrium consists of a strategy profile and a belief system such that the strategies are sequential rational given the belief system and the belief system is consistent given the strategy profile \cite{fudenberg1991perfect}. Based on the common information based belief system, \cite{ouyang2017dynamic} studied common information based perfect Bayesian equilibrium, and \cite{vasal2016systematic} studied structured perfect Bayesian equilibrium. Backward recursive formulas were given in both papers to find the corresponding perfect Bayesian equilibrium.

Bayesian games are also called games with incomplete information, and were studied in \cite{aumann1995repeated,rosenberg1998duality,zamir1992repeated}. This prior work mainly studied two-player zero-sum repeated Bayesian games, which are the most closely related to the present work. In two-player zero-sum repeated Bayesian games, minmax value, maxmin value, and game value are examined. Minmax value and maxmin value are also called the security level of the minimizer and the maximizer, respectively \cite{bacsar1998dynamic}. The strategy that guarantees the security level of the maximizer is called the security strategy of the maximizer, and the strategy assuring the security level of the minimizer is called the security strategy of the minimizer. When minmax value equals to maxmin value, we say the game has a value. In other words, there exists a Nash equilibrium, which is the security strategy pair. It was shown that the Nash equilibrium exists for both finite horizon and discounted two-player zero-sum repeated Bayesian games, but may not exist in infinite horizon average payoff two-player zero-sum repeated Bayesian games \cite{aumann1995repeated}. Later, \cite{rosenberg1998duality,sorin2002first} provided backward recursive formulas to compute the game value for finite horizon case and discounted case. In the same paper, dual games of two-player zero-sum repeated Bayesian games were also studied. Besides backward recursive formulas, it was shown that the security strategy of a player in the dual game with special initial parameters is also the player's security strategy in the primal game, and that the security strategy only depends on the sufficient statistics consisting of the belief on the player's own type, a real vector with the same size as that of the other player's type set, and the stage if this is a finite horizon game. The physical meaning of the real vector in the primal game was not clear in \cite{rosenberg1998duality,sorin2002first}. 

This paper adopts the same game model as the one used in \cite{aumann1995repeated,sorin2002first,zamir1992repeated,rosenberg1998duality}, and focuses on developing prescriptive methods for players, i.e. computing the security strategies of players. The main contribution of this paper includes three aspects. First, this paper clarifies that the real vector in a player's self-dependent sufficient statistics is the player's regret on the other player's type. Given the type $k$ of the other player, a player's regret on $k$ is the difference between the expected total payoff realized so far and the expected total payoff over all time using the security strategy if the other player is of type $k$. Second, explicit linear program formulations are provided to compute the initial condition of the self-dependent sufficient statistics and the security strategies. Third, in discounted cases when approximated security strategies are provided, the performance difference between the approximated security strategy and the security strategy is shown to be bounded.

This paper considers the following two-player zero-sum repeated Bayesian games. Each player has his own type that is only known to himself. The one-stage payoff function depends on both players' types and actions. At the beginning of the game, Nature chooses each player's type independently according to some publicly known probability, and send the type to the corresponding player. The players then choose their actions simultaneously based on their own types and both players' history actions at every stage. The reward of the maximizer is the sum of the one-stage payoffs all over the stages in a finite stage game, and the sum of discounted one-stage payoffs in a discounted game.

For the finite horizon case, we first provide an explicit linear program to compute players' security strategies based on the idea of realization plan in sequence form. Since the computed security strategy depends on both players' history actions, and the dimension of the history action space grows exponentially with the time horizon, prohibitive memory may be needed to record the strategies associated with the exponentially growing history action space as the time horizon increases. To save the memory, we provide the self-dependent sufficient statistics for players. Every player only needs to record the sufficient statistics whose size is time invariant, and computes the corresponding security strategy at every step. Based on the study in \cite{rosenberg1998duality}, this paper shows that the sufficient statistics consists of the player's belief on his/her own type and the player's regret on the other player's type. Besides, explicit linear programs are presented to compute the initial regrets in the game, and the security strategies based on the self-dependent sufficient statistics. Simulation results show that the two security strategies developed from different methods achieve the same game value.

For the discounted case, the self-dependent sufficient statistics consists of the belief on a player's own type and the \emph{anti-discounted} regret on the other player's type. An approximated security strategy that depends on the self-dependent sufficient statistics is provided, and a linear program is given to compute the approximated security strategy. Moreover, the performance difference between the approximated security strategy and the security strategy is studied, and a bound on the performance difference is presented.

The remainder of this paper is organized as follows. Section II presents the main results for finite horizon games. Section III discusses discounted games. Section IV demonstrates the main results on a jamming problem in underwater sensor networks. Finally, section V provides some future work.

\section{$T$-Stage Repeated Bayesian Games}
Let $\mathbb{R}^n$ denote the $n$-dimensional real space. For a finite set $K$, $|K|$ and $\Delta(K)$ denotes its cardinality and the set of probability distributions over $K$, respectively. Symbols $\mathbf{1}$ and $\mathbf{0}$ denote appropriately dimensional column vectors whose elements are all $1$ and $0$, respectively. Let $v(0),v(1),\ldots$ be a sequence of real values. We adopt the convention that $\sum_{t=1}^0 v(t)=0$, and $\prod_{t=1}^0 v(t)=1$. The supremum norm of a function $f: D\rightarrow \mathbb{R}$ is defined as $\|f\|_{\sup}=\sup_{x\in D}|f(x)|$, where $D$ is a non-empty set.

A two-player zero-sum repeated Bayesian game is specified by the seven-tuple $(K,L,A,B,M,p_0,q_0)$, where
\begin{itemize}
  \item $K$ and $L$ are non-empty finite sets, called player 1 and 2's type sets, respectively.
  \item $A$ and $B$ are non-empty finite sets, called player 1 and 2's action sets, respectively.
  \item $M: K\times L \times A \times B \rightarrow \mathbb{R}$ is the one-stage payoff function. $M^{kl}$ indicates the payoff matrix given player 1's type $k\in K$ and player 2's type $l\in L$. The element $M^{kl}_{a,b}$ of matrix $M^{kl}$, also denoted as $M(k,l,a,b)$, is the payoff given player 1's type $k\in K$ and action $a\in A$, and player 2's type $l\in L$ and action $b\in B$.
  \item $p_0\in \Delta(K)$ and $q_0\in \Delta(L)$ are the initial probabilities on $K$ and $L$, respectively. Without loss of generality, we assume $p_0^k,q_0^l>0$ for any $k\in K$ and $l\in L$.
\end{itemize}

A $T$-stage repeated Bayesian game is played as follows. Let $a_t,\ b_t$ denote player 1 and player 2's actions at stage $t=1,\ldots,T$, respectively. At stage $1$, $k$ and $l$ are chosen independently according to $p_0$ and $q_0$, and communicated to player 1 and 2, respectively. After the types are chosen, at stage $t=1,\ldots,T$, each player chooses his action independently, and announces it publicly. The payoff of player 1 at stage $t$ is $M(k,l,a_t,b_t)$. At stage $t=1,\ldots,T$, player 1 and 2's history action sequences $h_t^A$ and $h_t^B$ are defined as $h_t^A=(a_1,\ldots,a_{t-1})$ and $h_t^B=(b_1,\ldots,b_{t-1})$, and their history action spaces are defined as $H^A_{t}=A^{t-1}$ and $H^B_t=B^{t-1}$, respectively. We assume that $H_1^A,H_1^B,H_0^A,H_0^B=\emptyset$. With a little abuse of the terminology $\in$, we use $a_s\in h_t^A$ and $h_s^A\in h_t^A$ to indicate $a_s$ and $h_s^A$ are player 1's action and history action sequence at stage $s$ in the history action sequence $h_t^A$ for any $s=1,\ldots,t-1$. Similarly, $b_s\in h_t^B$ and $h_s^B\in h_t^B$ means that $b_s$ and $h_s^B$ are player 2's action and history action sequence at stage $s$ in the history action sequence $h_t^B$ for any $s=1,\ldots,t-1$.

A behavior strategy for player 1 is an element of $\sigma=(\sigma_t)_{t=1}^T$, where $\sigma_t$ is a map from $K\times H^A_t \times H^B_t$ to $\Delta(A)$. Similarly, a behavior strategy for player 2 is an element of $\tau=(\tau_t)_{t=1}^T$, where $\tau_t$ is a map from $L\times H^A_t \times H^B_t$ to $\Delta(B)$. Denote by $\Sigma$ and $\mathcal{T}$ the sets of strategies of player 1 and 2, respectively. Denote by $\sigma_t^{a_t}(\cdot,\cdot,\cdot)$ and $\tau_t^{b_t}(\cdot,\cdot,\cdot)$ the probabilities of playing $a_t$ and $b_t$ at stage $t$, respectively.

A quadruple $(p_0,q_0,\sigma,\tau)$ induces a probability distribution $P_{p_0,q_0,\sigma,\tau}$ on the set $\Omega=K\times L\times (A\times B)^T$ of plays. $\mathbb{E}_{p_0,q_0,\sigma,\tau}$ stands for the corresponding expectation. The payoff with initial probabilities $p_0,q_0$ and strategies $\sigma,\tau$ of the $T$-stage repeated Bayesian game is defined as $\gamma_T(p_0,q_0,\sigma,\tau)=\mathbb{E}_{p_0,q_0,\sigma,\tau}\left(\sum_{t=1}^{T}M(k,l,a_t,b_t) \right).$

The $T$-stage game $\Gamma_T(p_0,q_0)$ is defined as a two-player zero-sum repeated Bayesian game equipped with initial distribution $p_0$ and $q_0$, strategy spaces $\Sigma$ and $\mathcal{T}$, and payoff function $\gamma_T(p_0,q_0,\sigma,\tau)$. In this game, player 1 wants to \emph{maximize} the payoff $\gamma_T(p_0,\sigma,\tau)$, while player 2 wants to \emph{minimize} it.

Consider a $T$-stage game $\Gamma_T(p_0,q_0)$. The security level $\underline{V}_T(p_0,q_0)$ of player 1 is defined as $\underline{V}_T(p_0,q_0)=\max_{\sigma\in\Sigma}\min_{\tau\in\mathcal{T}}\gamma_T(p_0,q_0,\sigma,\tau)$, and the strategy $\sigma^*\in \Sigma$ which achieves player 1's security level is called the security strategy of player 1. Similarly, the security level $\overline{V}_T(p_0,q_0)$ of player 2 is defined as $\overline{V}_T(p_0,q_0)=\min_{\tau\in\mathcal{T}}\max_{\sigma\in\Sigma}\gamma_T(p_0,q_0\sigma,\tau)$, and the strategy $\tau^*\in \mathcal{T}$ which achieves player 2's security level is called the security strategy of player 2. When $\underline{V}_T(p_0,q_0)=\overline{V}_T(p_0,q_0)$, we say game $\Gamma_T(p_0,q_0)$ has a value, i.e. the game has a Nash equilibirum. Since game $\Gamma_T(p_0,q_0)$ is a finite game, it always has a value denoted by $V_T(p_0,q_0)$ \cite{sorin2002first}.

\subsection{LP formulations for players' security strategies}
A $T$-stage Bayesian repeated game is a finite game, and its security strategy can be computed by solving a linear program based on the sequence form \cite{von1996efficient}. The linear program provided in \cite{von1996efficient}, however, can not be directly used, because in our case, the strategies of both players depend on their own types which is not the same situation as in \cite{von1996efficient}. Therefore, we adopt the idea of \emph{realization plan} in the sequence form, and construct an explicit linear program for $T$-stage Bayesian repeated games.

Let us first introduce the \emph{realization plan}. Define player 1 and 2's realization plan $x_{k,h_t^A,h_t^B}^{a_t}$ and $y_{l,h_t^A,h_t^B}^{b_t}$ as
\begin{align}
  x_{k,h_t^A,h_t^B}^{a_t}=p^k\prod_{s=1}^t \sigma_s^{a_s}(k,h_s^A,h_s^B), \forall t=0,\ldots,T\label{eq: x}\\
  y_{l,h_t^A,h_t^B}^{b_t}=q^l\prod_{s=1}^t \tau_s^{b_s}(l,h_s^A,h_s^B),\forall t=0,\ldots,T\label{eq: y}
\end{align}
where $a_s,h_s^A\in h_t^A$ and $b_s,h_s^A\in h_t^A$ for all $s=1,\ldots,t-1$. It is easy to verify that the joint probability $P(k,l,h_{t+1}^A,h_{t+1}^B)$ satisfies $P(k,l,h_{t+1}^A,h_{t+1}^B)=x_{k,h_t^A,h_t^B}^{a_t}y_{l,h_t^A,h_t^B}^{b_t}$, where $a_t,h_t^A\in h_{t+1}^A$ and $b_t,h_{t}^B\in h_{t+1}^B$. Let $x_{t}=(x_{k,h_t^A,h_t^B})_{k\in K,h_t^A\in H_t^A,h_t^B\in H_t^B}$ and $y_{t}=(y_{l,h_t^A,h_t^B})_{l\in L,h_t^A\in H_t^A,h_t^B\in H_t^B}$. Denote by $x=(x_{t})_{t=1}^T$ and $y=(y_t)_{t=1}^T$ player 1 and 2's realization plans over all the $T$-stage Bayesian game. Player 1's realization plan $x$ satisfies constraint (\ref{eq: x constraint   1}-\ref{eq: x constraint 3}), and the corresponding set is denoted by $X$. Similarly, player 2's realization plan $y$ satisfies constraint (\ref{eq: y constraint   1}-\ref{eq: y constraint 3}), and the corresponding set is denoted by $Y$.
\begin{align}
  \mathbf{1}^T x_{k,h_t^A,h_t^B}=&x_{k,h_{t-1}^A,h_{t-1}^B}^{a_{t-1}}, \label{eq: x constraint 1}\\
  x_{k,h_t^A,h_t^B} \geq &\mathbf{0},\label{eq: x constraint 3}
\end{align}
for all $t=1,\ldots,T$, $k\in K$, $h_t^A\in H_t^A$, and $h_t^B\in H_t^B.$
\begin{align}
  \mathbf{1}^T y_{l,h_t^A,h_t^B}=&y_{l,h_{t-1}^A,h_{t-1}^B}^{b_{t-1}}, \label{eq: y constraint 1}\\
  y_{l,h_t^A,h_t^B}\geq &\mathbf{0},\label{eq: y constraint 3}
\end{align}
for all $t=1,\ldots,T$, $l\in L$, $h_t^A\in H_t^A$, and $h_t^B\in H_t^B.$

With perfect recall, for either player, looking for a security strategy is the same as looking for a realization plan that achieves the security level of the player \cite{von1996efficient}.

Given player 1's realization plan, define player 1's weighted future security payoff $u^{a_t,b_t}_{l,h_t^A,h_t^B}(x)$ for $t=0,\ldots,T$ as
\begin{align}
  &u^{a_t,b_t}_{l,h_t^A,h_t^B}(x)=\min_{\tau_{t+1:T}(l)\in \mathcal{T}_{t+1:T}(l)} \sum_{k\in K}x_{k,h_t^A,h_t^B}^{a_t} \nonumber\\
  &\mathbb{E}\left(\sum_{s=t+1}^T M(k,l,a_s,b_s)|k,l,h_{t+1}^A,h_{t+1}^B\right), \label{eq: u}
\end{align}
where $h_{t+1}^A=(h_t^A,a_t)$, $h_{t+1}^B=(h_t^B,b_t)$, $\tau_{t+1:T}(l)=(\tau_s(l,:,:))_{s=t+1}^T$, and $\mathcal{T}_{t+1:T}(l)$ is the set of player 2's behavior strategies from $t+1$ to $T$ given player 2's type $l$. The pairs $(h_t^A,a_t)$ and $(h_t^B,b_t)$ indicate concatenation. Similarly, define player 2's weighted future security payoff $w^{a_t,b_t}_{k,h_t^A,h_t^B}(y)$ as
\begin{align}
  &w^{a_t,b_t}_{k,h_t^A,h_t^B}(y)=\max_{\sigma_{t+1:T}(k)\in \Sigma_{t+1:T}(k)} \sum_{l\in L}y_{l,h_t^A,h_t^B}^{b_t}\nonumber\\
  &\mathbb{E}\left(\sum_{s=t+1}^T M(k,l,a_s,b_s)|k,l,h_{t+1}^A,h_{t+1}^B\right), \label{eq: w}
\end{align}
for $t=0,\ldots,T$ where $\sigma_{t+1:T}(k)=(\sigma_s(k,:,:))_{s=t+1}^T$, and $\Sigma_{t+1:T}(k)$ is the set of player 1's strategies from stage $t+1$ to $T$ given player 1's type $k\in K$.

For $t=1,\ldots,T$, $u_{l,h_t^A,h_t^B}(x),w_{k,h_t^A,h_t^B}(y)$ are $|A|\times |B|$ matrices whose elements are $u_{l,h_t^A,h_t^B}^{a_t,b_t}(x)$ and $w_{k,h_t^A,h_t^B}^{a_t,b_t}(y)$, respectively. For $t=0$, since $a_t,b_t,h_t^A,h_t^B\in \emptyset$, $u_{l,h_t^A,h_t^B}(x)$ and $w_{k,h_t^A,h_t^B}(y)$ are scalars, and denoted as $u_{l,0}(x)$ and $w_{k,0}(y)$, respectively. Define $u_t(x)=(u_{l,h_t^A,h_t^B}(x))_{l\in L,h_t^A\in H_t^A,h_t^B\in H_t^B}$, and $u(x)=(u_t(x))_{t=1}^{T-1}$. Similarly, define $w_t(y)=(w_{k,h_t^A,h_t^B}^{a_t,b_t}(y))_{k\in K,h_t^A\in H_t^A,h_t^B\in H_t^B}$, and $w(y)=(w_t(y))_{t=1}^{T-1}$. For the convenience of the rest of this paper, let $U$ and $W$ be the real spaces of appropriate dimensions which $u$ and $w$ take values in. The weighted future security payoffs $u,w$ satisfy backward recursive formulas.
\begin{lemma}
  \label{lemma: weighted future security payoff} Consider a $T$-stage Bayesian game $\Gamma_T(p,q)$.
Player 1 and 2's weighted future security payoffs $u_{l,h_t^A,h_t^B}^{a_t,b_t}(x)$ and $w_{k,h_t^A,h_t^B}^{a_t,b_t}(y)$ defined in (\ref{eq:   u}) and (\ref{eq:   w}) satisfy
\begin{align}
 &u_{l,h_t^A,h_t^B}^{a_t,b_t}(x)= \min_{\tau_{t+1}(l,h_{t+1}^A,h_{t+1}^B)\in \Delta(B)}\left(\sum_{k\in K}x_{k,h_{t+1}^A,h_{t+1}^B}^TM^{kl}\right. \nonumber\\
 &\left.+\mathbf{1}^Tu_{l,h_{t+1}^A,h_{t+1}^B}(x)\right)
 \tau_{t+1}(l,h_{t+1}^A,h_{t+1}^B),\label{eq: u recursive} \\
 &w_{k,h_t^A,h_t^B}^{a_t,b_t}(y)= \max_{\sigma_{t+1}(k,h_{t+1}^A,h_{t+1}^B)\in \Delta(A)}\sigma_{t+1}(k,h_{t+1}^A,h_{t+1}^B)^T\nonumber\\
 &\left(\sum_{l\in L}M^{kl}y_{l,h_{t+1}^A,h_{t+1}^B}+w_{k,h_{t+1}^A,h_{t+1}^B}(y)\mathbf{1}\right),\label{eq: w recursive}
\end{align}
for all $t=0,\ldots,T-1$, where $h_{t+1}^A=(h_t^A,a_t)$, $h_{t+1}^B=(h_t^A,b_t)$, $x_{k,h_{t+1}^A,h_{t+1}^B}\in \mathbb{R}^{|A|}$ is player 1's realization plan whose element is defined as in (\ref{eq: x}), and $y_{l,h_{t+1}^A,h_{t+1}^B}\in \mathbb{R}^{|B|}$ is player 2's realization plan whose element is defined as in (\ref{eq: y}). Here, $(h_t^A,a_t)$ and $(h_t^A,b_t)$ indicate concatenation.
\end{lemma}
\begin{proof}
According to equation (\ref{eq:   u}), we have
\begin{align*}
  &u_{l,h_{T-1}^A,h_{T-1}^B}^{a_{T-1},b_{T-1}}(x)\\
  =&\min_{\tau_T(l,h_T^A,h_T^B)\in \Delta(B)}\sum_{k\in K}x_{k,h_{T-1}^A,h_{T-1}^B}^{a_{T-1}}\sigma_T^T(k,h_T^A,h_T^B)M^{kl}\\
  &\tau_T(l,h_T^A,h_T^B)\\
  =&\min_{\tau_T(l,h_T^A,h_T^B)\in \Delta(B)}\sum_{k\in K}x_{k,h_{T}^A,h_{T}^B}^TM^{kl}\tau_T(l,h_T^A,h_T^B) \\
  =&\min_{\tau_T(l,h_T^A,h_T^B)\in \Delta(B)}\left(\sum_{k\in K}x_{k,h_{T}^A,h_{T}^B}^TM^{kl}+\mathbf{1}^Tu_{l,h_{T}^A,h_{T}^B}(x)\right)\\
  &\tau_T(l,h_T^A,h_T^B).
\end{align*}
The last equality holds because $u_{l,h_{T}^A,h_{T}^B}(x)$ is a zero matrix.

Suppose equation (\ref{eq: u recursive}) holds for all $t=1,\ldots,T-1$. Consider the case of $t-1$.
\begin{align*}
  &u_{l,h_{t-1}^A,h_{t-1}^B}^{a_{t-1},b_{t-1}}(x)\\
  =&\min_{\tau_{t:T}(l)\in\mathcal{T}_{t:T}(l)}\sum_{k\in K}x_{k,h_{t-1}^A,h_{t-1}^B}^{a_{t-1}}\left(\sigma^T_t(k,h_t^A,h_t^B)M^{kl}\right.\\
  & \left.\tau_t(l,h_t^A,h_t^B)+\sum_{a_t\in A}\sum_{b_t\in B}P(a_t,b_t|k,l,h_t^A,h_t^B)\right.\\
  &\left.\mathbb{E}(\sum_{s=t+1}^TM(k,l,a_s,b_s)|k,l,h_t^A,h_t^B,a_t,b_t)\right)\\
  =& \min_{\tau_t(l,h_t^A,h_t^B)\in\Delta(B)}\left\{\sum_{k\in K}x_{k,h_t^A,h_t^B}^TM^{kl}\tau_t(l,h_t^A,h_t^B)\right.\\
  &\left.+\sum_{a_t\in A}\sum_{b_t\in B} \tau_t^{b_t}(l,h_t^A,h_t^B)\min_{\tau_{t+1:T}(l)\in\mathcal{T}_{t+1:T}(l)} \sum_{k\in K} x_{k,h_t^A,h_t^B}^{a_t}\right.\\
  &\left.\mathbb{E}\left(\sum_{s=t+1}^TM(k,l,a_s,b_s)|k,l,h_{t+1}^A,h_{t+1}^B\right)\right\}\\
  =& \min_{\tau_t(l,h_t^A,h_t^B)\in\Delta(B)}\left(\sum_{k\in K}x_{k,h_t^A,h_t^B}^TM^{kl}+\mathbf{1}^T u_{l,h_t^A,h_t^B}(x)\right)\\
  &\tau_t(l,h_t^A,h_t^B).
\end{align*}
Therefore, equation (\ref{eq: u recursive}) holds for all $t=0,\ldots,T-1$.

Following the same steps, equation (\ref{eq: w recursive}) can be shown.
\end{proof}

Now, Let us present the explicit LP formulations. 
\begin{theorem}
\label{theorem: LP primal game T stage}
Consider a $T$-stage repeated Bayesian game $\Gamma_T(p,q)$. The game value $V_T(p,q)$ satisfies
\begin{align}
  &V_T(p,q)=\max_{x\in X,u\in U,u_{:,0}\in\mathbb{R}^{|L|}}\sum_{l\in L} q^l u_{l,0} \label{eq: LP player 1}\\
  s.t. & \sum_{k\in K}{M^{kl}}^Tx_{k,h_1^A,h_1^B}+{u_{l,h_1^A,h_1^B}}^T\mathbf{1} \geq u_{l,0} \mathbf{1},\forall l\in L, \label{eq: LP player 1-1}\\
  & \sum_{k\in K}{M^{kl}}^Tx_{k,h_{t+1}^A,h_{t+1}^B}+{u_{l,h_{t+1}^A,h_{t+1}^B}}^T\mathbf{1} \geq u_{l,h_{t}^A,h_{t}^B}^{a_{t},b_{t}}\mathbf{1}, \nonumber\\
  &\forall t=1,\ldots,T-1,l\in L,h_t^A\in H_t^A,h_t^B\in H_t^B, \label{eq: LP plyaer 1-2}
\end{align}
where $u_{l,h_T^A,h_T^B}$ is a zero matrix for all $l\in L$, $X$ is a set including all real vectors satisfying (\ref{eq: x constraint 1}-\ref{eq: x constraint 3}), and $U$ is a real space of appropriate dimension. Player 1's security strategy $\sigma^{a_t*}_t(k,h_t^A,h_t^B)$ for all $t=1,\ldots,T$,
$k\in K$, $h_t^A\in H_t^A$, $h_t^B\in H_t^B$, and $ a_t\in A$ satisfies
\begin{align}
{\sigma^{a_t}_t}^*(k,h_t^A,h_t^B)={x_{k,h_t^A,h_t^B}^{a_t*}}/{x_{k,h_{t-1}^A,h_{t-1}^B}^{a_{t-1}*}}. \label{eq: player 1's security strategy}
\end{align}

Dually, the game value $V_T(p,q)$ also satisfies
\begin{align}
  &V_T(p,q)=\min_{y\in Y,w\in W,w_{:,0}\in\mathbb{R}^{|K|}}\sum_{k\in K} p^k w_{k,0} \label{eq: LP player 2}\\
  s.t. & \sum_{l\in L} M^{kl} y_{l,h_1^A,h_1^B}+w_{k,h_1^A,h_1^B}\mathbf{1} \leq w_{k,0}\mathbf{1},\forall k\in K, \label{eq: LP player 2 1}\\
  & \sum_{l\in L} M^{kl} y_{l,h_{t+1}^A,h_{t+1}^B}+w_{k,h_{t+1}^A,h_{t+1}^B}\mathbf{1} \leq w_{k,h_{t}^A,h_{t}^B}^{a_{t},b_{t}}\mathbf{1},\nonumber\\
  &\forall t=1,\ldots,T-1,k\in K,h_t^A\in H_t^A,h_t^B\in H_t^B, \label{eq: LP player 2 2}
\end{align}
where $w_{k,h_T^A,h_T^B}$ is a zero matrix for all $k\in K$, $Y$ is a set including all real vectors satisfying (\ref{eq: y constraint 1}-\ref{eq: y constraint 3}), and $W$ is a real space of appropriate dimension. Player 2's security strategy $\tau^{b_t*}_t(l,h_t^A,h_t^B)$ for all $t=1,\ldots,T$,
$l\in L$, $h_t^A\in H_t^A$, $h_t^B\in H_t^B$, and $ a_t\in A$ satisfies
\begin{align}
\tau^{b_t*}_t(l,h_t^A,h_t^B)={y_{l,h_t^A,h_t^B}^{b_t*}}{y_{l,h_{t-1}^A,h_{t-1}^B}^{b_{t-1}*}}.\label{eq: player 2's security strategy}
\end{align}
\end{theorem}
\begin{proof}
Equation (\ref{eq:   u}) indicates that $V_T(p,q)=\max_{x\in X}\sum_{l\in L}q^lu_{l,0}(x)$, where $u_{l,0}(x)$ satisfies (\ref{eq: u recursive}).

According to the duality theory in LP problem, equation (\ref{eq: u recursive}) can be rewritten as
\begin{align}
&u_{l,h_t^A,h_t^B}^{a_t,b_t}(x)=\max_{u_{l,h_t^A,h_t^B}^{a_t,b_t}\in \mathbb{R}} u_{l,h_t^A,h_t^B}^{a_t,b_t} \label{eq: u(x) tend 0}\\
s.t.&\sum_{k\in K}{M^{kl}}^Tx_{k,h_{t+1}^A,h_{t+1}^B}+u_{l,h_{t+1}^A,h_{t+1}^B}^T(x)\mathbf{1} \geq u_{l,h_{t}^A,h_{t}^B}^{a_{t},b_{t}}\mathbf{1}, \nonumber\\
& \forall t=0,\ldots,T-1. \label{eq: u(x) tend 1}
\end{align}

For $t=T-1$, since $u_{l,h_T^A,h_T^B}(x)$ is a zero matrix, we have
\begin{align}
u_{l,h_{T-1}^A,h_{T-1}^B}^{a_{T-1},b_{T-1}}(x)&=\max_{u_{l,h_{T-1}^A,h_{T-1}^B}^{a_{T-1},b_{T-1}}\in \mathbb{R}} u_{l,h_{T-1}^A,h_{T-1}^B}^{a_{T-1},b_{T-1}} , \label{eq: u(x) tend 2}\\
s.t.&\sum_{k\in K}{M^{kl}}^Tx_{k,h_T^A,h_T^B} \geq u_{l,h_{T-1}^A,h_{T-1}^B}^{a_{T-1},b_{T-1}}\mathbf{1}. \label{eq: u(x) tend 3}
\end{align}

For $t=T-2$, we define
\begin{align}
  &\hat{u}_{l,h_{T-2}^A,h_{T-2}^B}^{a_{T-2},b_{T-2}}(x)\\
  =&\max_{\substack{u_{l,h_{T-2}^A,h_{T-2}^B}^{a_{T-2},b_{T-2}}\in \mathbb{R}\\u_{l,(h_{T-2}^A,a_{T-2}),(h_{T-2}^B,b_{T-2})}\in \mathbb{R}^{|A|\times|B|}}} u_{l,h_{T-2}^A,h_{T-2}^B}^{a_{T-2},b_{T-2}} \label{eq: u(x) tend 4}\\
  s.t.& \sum_{k\in K} {M^{kl}}^T x_{k,h_{T-1}^A,h_{T-1}^B} + u_{l,(h_{T-2}^A,a_{T-2}),(h_{T-2}^B,b_{T-2})}^T \mathbf{1} \nonumber\\
   &\geq u_{l,h_{T-2}^A,h_{T-2}^B}^{a_{T-2},b_{T-2}} \mathbf{1}, \label{eq: u(x) tend 5}\\
  & \sum_{k\in K} {M^{kl}}^T x_{k,h_T^A,h_T^B} \geq u_{l,(h_{T-2}^A,a_{T-2}),(h_{T-2}^B,b_{T-2})}^{a_{T-1},b_{T-1}} \mathbf{1}, \nonumber \\
  &\forall a_{T-1}\in A,b_{T-1}\in B. \label{eq: u(x) tend 6}
\end{align}
We will show that $u_{l,h_{T-2}^A,h_{T-2}^B}^{a_{T-2},b_{T-2}}(x)=\hat{u}_{l,h_{T-2}^A,h_{T-2}^B}^{a_{T-2},b_{T-2}}(x)$. Equation (\ref{eq: u(x) tend 0}) implies that
\begin{align}
 &u_{l,h_{T-2}^A,h_{T-2}^B}^{a_{T-2},b_{T-2}}(x)=\max_{u_{l,h_{T-2}^A,h_{T-2}^B}^{a_{T-2},b_{T-2}}\in \mathbb{R}} u_{l,h_{T-2}^A,h_{T-2}^B}^{a_{T-2},b_{T-2}} \label{eq: u(x) tend 7}\\
  s.t.& \sum_{k\in K} {M^{kl}}^T x_{k,h_{T-1}^A,h_{T-1}^B} + u_{l,(h_{T-2}^A,a_{T-2}),(h_{T-2}^B,b_{T-2})}^{*T} \mathbf{1} \nonumber\\
  &\geq u_{l,h_{T-2}^A,h_{T-2}^B}^{a_{T-2},b_{T-2}} \mathbf{1}, \label{eq: u(x) tend 8}
\end{align}
where the element in $u_{l,(h_{T-2}^A,a_{T-2}),(h_{T-2}^B,b_{T-2})}^{*}$ is the corresponding maximum of LP (\ref{eq: u(x) tend 2}-\ref{eq: u(x) tend 3}).

Let ${u_{l,h_{T-2}^A,h_{T-2}^B}^{a_{T-2},b_{T-2}}}^\star$, $u_{l,(h_{T-2}^A,a_{T-2}),(h_{T-2}^B,b_{T-2})}^{\star}$ be the optimal solution to LP problem (\ref{eq: u(x) tend 4}-\ref{eq: u(x) tend 6}). Since $u_{l,(h_{T-2}^A,a_{T-2}),(h_{T-2}^B,b_{T-2})}^{\star}$ satisfies equation (\ref{eq: u(x) tend 6}) and hence (\ref{eq: u(x) tend 3}), we have ${u_{l,(h_{T-2}^A,a_{T-2}),(h_{T-2}^B,b_{T-2})}^{a_{T-1},b_{T-1}}}^*\geq {u_{l,(h_{T-2}^A,a_{T-2}),(h_{T-2}^B,b_{T-2})}^{a_{T-1},b_{T-1}}}^\star$ for any $a_{T-1}\in A,$ $b_{T-1}\in B$. Together with equation (\ref{eq: u(x)    tend 5}), we show that ${u_{l,h_{T-2}^A,h_{T-2}^B}^{a_{T-2},b_{T-2}}}^\star$ satisfies equation (\ref{eq: u(x)   tend 8}). Therefore, ${u_{l,h_{T-2}^A,h_{T-2}^B}^{a_{T-2},b_{T-2}}}^\star$, $u_{l,(h_{T-2}^A,a_{T-2}),(h_{T-2}^B,b_{T-2})}^{\star}$ is a feasible solution to the nested LP problem (\ref{eq: u(x) tend  7}-\ref{eq: u(x)   tend 8}), and $\hat{u}_{l,h_{T-2}^A,h_{T-2}^B}^{a_{T-2},b_{T-2}}(x)\leq u_{l,h_{T-2}^A,h_{T-2}^B}^{a_{T-2},b_{T-2}}(x).$

Meanwhile, let ${u_{l,h_{T-2}^A,h_{T-2}^B}^{a_{T-2},b_{T-2}}}^*, u_{l,(h_{T-2}^A,a_{T-2}),(h_{T-2}^B,b_{T-2})}^{*}$ be the optimal solution to the nested LP (\ref{eq: u(x) tend 7}-\ref{eq: u(x) tend 8}). It is easy to check that ${u_{l,h_{T-2}^A,h_{T-2}^B}^{a_{T-2},b_{T-2}}}^*, u_{l,(h_{T-2}^A,a_{T-2}),(h_{T-2}^B,b_{T-2})}^{*}$ is a feasible solution to LP (\ref{eq: u(x) tend 4}-\ref{eq: u(x) tend 6}), and hence $u_{l,h_{T-2}^A,h_{T-2}^B}^{a_{T-2},b_{T-2}}(x)\leq \hat{u}_{l,h_{T-2}^A,h_{T-2}^B}^{a_{T-2},b_{T-2}}(x).$ Therefore,  $u_{l,h_{T-2}^A,h_{T-2}^B}^{a_{T-2},b_{T-2}}(x)= \hat{u}_{l,h_{T-2}^A,h_{T-2}^B}^{a_{T-2},b_{T-2}}(x).$

Following the same steps, we can show the case for $t=T-3,\ldots,0$, and have
\begin{align}
  &u_{l,0}(x)=\max_{u_{l,0}\in \mathbb{R},u\in U}u_{l,0} \label{eq: u(x)}\\
  s.t.& \sum_{k\in K}{M^{kl}}^Tx_{k,h_1^A,h_1^B}+u_{l,h_1^A,h_1^B}^T\mathbf{1} \geq u_{l,0} \mathbf{1}, \label{eq: u(x) 1}\\
  &\sum_{k\in K}{M^{kl}}^Tx_{k,h_{t+1}^A,h_{t+1}^B}+u_{l,h_{t+1}^A,h_{t+1}^B}^T\mathbf{1} \geq u_{l,h_{t}^A,h_{t}^B}^{a_{t},b_{t}}\mathbf{1}, \nonumber\\
  &\forall t=1,\ldots,T-1, h_t^A\in H_t^A, h_t^B\in H_t^B,\label{eq: u(x) 2}
\end{align}
and equation (\ref{eq: LP player   1}-\ref{eq: LP plyaer 1-2}) is shown. From the definition of $x$ in (\ref{eq: x}), we derive player 1's security strategy as in (\ref{eq: player 1's security strategy}).

Following the same steps, we show equation (\ref{eq: LP player   2}-\ref{eq: LP player 2 2}) is true, and player 2's security strategy is computed as in (\ref{eq: player 2's security strategy}).
\end{proof}

Notice that the sizes of the LP formulations in (\ref{eq: LP player   1}-\ref{eq: LP plyaer 1-2}) and (\ref{eq: LP player   2}-\ref{eq: LP player 2 2}) are both linear in the size of the game tree, i.e. linear in the sizes of both players' type sets, polynomial in the sizes of both players' action sets, and exponential in the time horizon. 

\subsection{Security strategies based on fixed-sized sufficient statistics and dual games}
Theorem \ref{theorem: LP primal game T stage} provides LP formulations to compute both players' security strategies which depend on both players' history actions. Notice that history action space grows exponentially on time horizon which makes this LP formulation undesirable as the horizon length grows. As time horizon gets long, players need a great amount of memories to record players' history action space and the corresponding security strategy. In order to remedy this drawback, we now consider another type of security strategies, which depend on fixed-sized sufficient statistics to save memories.

Our starting point is a result of \cite{de1996repeated,sorin2002first}, which showed that a player's security strategy in the dual game with some special initial parameters is also the player's security strategy in the primal game, and the security strategy only depends on a fixed-sized sufficient statistics. This subsection clarifies what the special initial parameters in dual games mean in the primal game, and give LP formulation and algorithms to compute the initial parameters and the corresponding security strategies.

First of all, we would like to introduce \emph{two dual games} of a $T$-stage repeated Bayesian game $\Gamma_T(p,q)$. Game $\Gamma_T(p,q)$'s type 1 dual game $\tilde{\Gamma}^1_T(\mu,q)$ is defined with respect to its first parameter $p$, where $\mu\in \mathbb{R}^{|K|}$ is called the initial regret with respect to player 1's type. The dual game $\tilde{\Gamma}^1_T(\mu,q)$ is played as follows. Player 1 chooses $k$ without informing player 2. Independently, nature chooses player 2's type according to $q$, and announces it to player 2 only. From stage $1$ to $T$, knowing both players' history actions, both players choose actions simultaneously. Let $p$ be player 1's strategy to choose his own type, and $\sigma\in\Sigma$ and $\tau\in \mathcal{T}$ be player 1 and 2's strategies to choose actions. Player 1's payoff $\tilde{\gamma}^1_T(\mu,q,p,\sigma,\tau)$ is defined as $\tilde{\gamma}^1_T(\mu,q,p,\sigma,\tau)=\mathbb{E}_{p,q,\sigma,\tau}\left(\mu^k+\sum_{t=1}^T M(k,l,a_t,b_t) \right)$.
We can see that the main difference between the type 1 dual game and the primal game is that in type 1 dual game, player 1 has an initial regret instead of a initial probability $p$, and he himself instead of the nature chooses his own type.

Similarly, the type 2 dual game $\tilde{\Gamma}_T^2(p,\nu)$ is defined with respect to the second parameter $q$, where $\nu\in\mathbb{R}^{|L|}$ is called the initial regret with respect to player 2's type. The dual game $\tilde{\Gamma}_T^2(p,\nu)$ is played as follows. Player 2 chooses $l$ without informing player 1. Meanwhile, player 1's type is chosen according to $p$, and is only announced to player 1. From stage $1$ to $T$, knowing both players' history actions, both players choose actions independently. Let $q$ be player 2's strategy to choose his type $l$. Player 1's payoff $\tilde{\gamma}^2_T(p,\nu,q,\sigma,\tau)$ is defined as $\tilde{\gamma}^2_T(p,\nu,q,\sigma,\tau)=\mathbb{E}_{p,q,\sigma,\tau}\left(\nu^l+\sum_{t=1}^T M(k,l,a_t,b_t)\right).$
In both dual games, player 1 wants to maximize the payoff, while player 2 wants to minimize it.

Both dual games are finite, and hence have game values denoted by $\tilde{V}_T^1(\mu,q)$ and $\tilde{V}^2_T(p,\nu)$. They are related to the game value of the primal game in the following way \cite{sorin2002first}.
\begin{align}
  \tilde{V}_T^1(\mu,q)=&\max_{p\in \Delta(K)}\{V_T(p,q)+p^T\mu\},\label{eq: game value relation 1, T stage}\\
  V_T(p,q)=&\min_{\mu\in\mathbb{R}^{|K|}}\{\tilde{V}_T^1(\mu,q)-p^T\mu\},\label{eq: game value relation 2, T stage}\\
  \tilde{V}_T^2(p,\nu)=&\min_{q\in \Delta(L)}\{V_T(p,q)+q^T\nu\},\label{eq: game value relation 3, T stage}\\
  V_T(p,q)=&\max_{\nu\in\mathbb{R}^{|L|}}\{\tilde{V}_T^2(q,\nu)-q^T\nu\}. \label{eq: game value relation 4, T stage}
\end{align}
Let $\mu^*$ and $\nu^*$ be the solutions to the optimal problems on the right hand side of (\ref{eq: game value relation 2,   T stage}) and (\ref{eq: game value relation 4,   T stage}), respectively. Player 2's security strategy in the type 1 dual game $\tilde{\Gamma}_T^1(\mu^*,q)$ is his security strategy in the primal game $\Gamma_T(p,q)$, and player 1's security strategy in the type 2 dual game $\tilde{\Gamma}_T^2(p,\nu^*)$ is also his security strategy in the primal game $\Gamma_T(p,q)$ \cite{rosenberg1998duality}.

The next questions are what $\mu^*$ and $\nu^*$ are, and how to compute them. To answer these questions, we have the following lemma.
\begin{lemma}
\label{lemma: optimal solution T stage}
Consider a $T$-stage repeated Bayesian game $\Gamma_T(p,q)$. Let $\sigma^*_{p,q}$ and $\tau^*_{p,q}$ be player 1 and 2's security strategies in $\Gamma_T(p,q)$, respectively. Denote by $x^*_{p,q}$ and $y^*_{p,q}$ the corresponding optimal realization plans of player 1 and 2. The optimal solution $\mu^*$ to the optimal problem $\min_{\mu\in\mathbb{R}^{|K|}}\{\tilde{V}_T^1(\mu,q)-p^T\mu\}$ is
\begin{align}
  \mu^{*k}=-w_{k,0}(y^*_{p,q}), \forall k\in K \label{eq: initial regret mu}
\end{align}
where $w_{k,0}(y^*_{p,q})=w_{k,h_0^A,h_0^B}^{a_0,b_0}(y^*_{p,q})$, which is defined in (\ref{eq: w}) and computed according to the linear program (\ref{eq: LP player 2}-\ref{eq: LP player 2 2}).

The optimal solution $\nu^*$ to the optimal problem $\max_{\nu\in\mathbb{R}^{|L|}}\{\tilde{V}_T^2(q,\nu)-q^T\nu\}$ is
\begin{align}
 \nu^{*l}=-u_{l,0}(x^*_{p,q}), \forall l\in L \label{eq: initial regret nu}
\end{align}
where $u_{l,0}(x^*_{p,q})=u_{l,h_0^A,h_0^B}^{a_0,b_0}(x^*_{p,q})$, which is defined in (\ref{eq: u}) and computed according to the linear program (\ref{eq: LP player   1}-\ref{eq: LP plyaer 1-2}).
\end{lemma}
\begin{proof}
  First, we prove that equation (\ref{eq: 0 game value 0}) is true.
  \begin{align}
  V_T(p,q)=p^T w_{:,0}(y^*_{p,q})=-p^T\mu^*.  \label{eq: 0 game value 0}
  \end{align}
  Equation (\ref{eq: w}) and (\ref{eq: y}) implies that
  \begin{align*}
    w_{k,0}(y^*_{p,q})=&\max_{\sigma(k)\in \Sigma(k)}\sum_{l\in L} q^l \mathbb{E}_{\sigma(k),y^*_l}\left(\sum_{s=1}^T M(k,l,
    a_s,b_s)|k,l\right)
  \end{align*}
  Therefore, we have
  \begin{align*}
    &\sum_{k\in K}p^kw_{k,0}(y^*_{p,q})\\
    =&\max_{\sigma\in \Sigma}\sum_{k\in K}\sum_{l\in L}p^kq^l\mathbb{E}_{\sigma(k),y^*_l}\left(\sum_{s=1}^T M(k,l,
    a_s,b_s)|k,l\right)\\
    =&\max_{\sigma\in \Sigma}\mathbb{E}_{\sigma,y^*}\left(\sum_{s=1}^T M(k,l,
    a_s,b_s)\right)\\
    =&\min_{y\in Y}\max_{\sigma\in \Sigma}\mathbb{E}_{\sigma,y}\left(\sum_{s=1}^T M(k,l,
    a_s,b_s)\right)\\
    =&\min_{\tau\in \mathcal{T}}\max_{\sigma\in \Sigma}\mathbb{E}_{\sigma,\tau}\left(\sum_{s=1}^T M(k,l,
    a_s,b_s)\right)=V_T(p,q)
  \end{align*}
  where the second equality holds because $y^*$ is player 2's security realization plan, and the last equality holds because of the perfect recall in this game \cite{von1996efficient}.

  Next, we show that
  \begin{align}
  \tilde{V}_T^1(\mu^*,q)=0. \label{eq: 0 game value 1}
  \end{align}

  Equation (\ref{eq: game value relation 1, T stage}) implies that $\tilde{V}_T^1(\mu^*,q)=\max_{p'\in \Delta(K)}\{V_T(p',q)+p'^T\mu^*\}\geq V_T(p,q)+p^T\mu^*=0$.

  Meanwhile, for any $p'\in \Delta(K)$, we have
  \begin{align*}
    &V_T(p',q)\\
    =    &\min_{y\in Y}\max_{\sigma\in \Sigma}\mathbb{E}_{p',q,\sigma,y}\left(\sum_{t=1}^T M(k,l,a_t,b_t)\right)\\
    &\leq \max_{\sigma\in\Sigma}\mathbb{E}_{p',q,\sigma,y^*_{p,q}}\left(\sum_{t=1}^T M(k,l,a_t,b_t)\right)\\
  =&\sum_{k\in K}p'^k \max_{\sigma(k)\in \Sigma(k)}\mathbb{E}_{q,\sigma(k),y^*_{p,q}}(\sum_{t=1}^T M(k,l,a_t,b_t)|k)\\
  =&\sum_{k\in K}p'^k w_{k,0}(y^*_{p,q})=-p'^T\mu^*,
  \end{align*}
  which implies that $V_T(p',q)+p'^T\mu^*\leq 0$ for any $p'\in \Delta(K)$. Hence, $\tilde{V}_T^1(\mu^*,q)\leq 0$ according to (\ref{eq: game value relation 1, T stage}). Therefore, equation (\ref{eq: 0 game value 1}) is true.

  Equation (\ref{eq: 0 game value 1}) implies that $\tilde{V}_T^1(\mu^*,q)-p^T\mu^*=-p^T\mu^*=V_T(p,q)$, where the second equality is based on (\ref{eq: 0 game value 0}). According to equation (\ref{eq: game value relation 2, T stage}), we see that $-w_{k,0}(y^*_{p,q})$ is an optimal solution to the optimal problem on the right hand side of (\ref{eq: game value relation 2, T stage}). From the proof of Theorem \ref{theorem: LP primal game T stage}, we see that $w_{:,0}(y^*_{p,q})=w_{:,0}^*$, where $w_{:,0}^*$ is the optimal solution to the linear program (\ref{eq: LP player   2}-\ref{eq: LP player 2 2}).

  Following the same steps, we show that
  \begin{align}
    \tilde{V}_T^2(q,\nu^*)=0,\label{eq: 0 game value 2}
  \end{align}
  and $-u_{l,0}(x^*_{p,q})$ is an optimal solution to the optimal problem on the right hand side of (\ref{eq: game value relation 4, T stage}). Moreover, $u_{:,0}(x^*_{p,q})$ equals to $u_{:,0}^*$ the optimal solution to the linear program (\ref{eq: LP player   1}-\ref{eq: LP plyaer 1-2}), according to the proof of Theorem \ref{theorem: LP primal game T stage}.
\end{proof}

In the primal game, the parameter $\mu^{*k}$ can be seen as player 2' initial regret given player 1's type $k$, i.e. the difference between the expected realized payoff before stage 1, which is $0$, and the expected total payoff using player 2's security strategy if player 1 is of type $k$. Parameter $\nu^{*l}$ is player 1's initial regret given player 2's type $l$, i.e. the difference between the expected realized payoff before stage 1 and the expected total payoff using player 1's security strategy if player 2 is of type $l$.
Now that we have figured out the two special parameters in the dual game, our next step is to study \emph{the security strategies in the dual games}. In type 1 dual game $\tilde{\Gamma}_T^1(\mu,q)$, player 2 keeps track of two variables, the \emph{belief state $q_t\in \Delta(L)$ on his own type}, and his \emph{regret $\mu_t\in \mathbb{R}^{|K|}$} on player 1's type. The \emph{belief on player 2's type} is defined as
\begin{align}
  q_t^l=P(l|k,h_t^A,h_t^B), \forall l\in L, t=1,\ldots,T,
\end{align}
and is updated as follows
\begin{align}
  q_{t+1}^l=&q^{+l}(b_t,z_t,q_t)=\frac{q_t^l z_t^l(b_t)}{\bar{z}_{q_t,z_t}(b_t)}, \forall l\in L,\hbox{with $q_1=q$,} \label{eq: belief state player 2}
\end{align}
where $z_t^l=\tau_t(l,h_t^A,h_t^B)\in \Delta(B)$, and $\bar{z}_{q_t,z_t}(b_t)=\sum_{l\in L}q_t^l z_t^l(b_t).$ The \emph{regret on player 1's type} is defined as
\begin{align*}
  \mu_t^k=\mu^k+\sum_{s=1}^{t-1}\mathbb{E}(M(k,l,a_s,b_s)|k,h_{s+1}^A,h_{s+1}^B),
\end{align*}
for all $k\in K,t=1,\ldots,T,$
and is updated as follows
\begin{align}
  \mu_{t+1}^k=& \mu^{+k}(\mu_t,a_t,b_t,z_t,q_t) \nonumber\\
  =&\mu^k_t+\mathbb{E}(M(k,l,a_t,b_t)|k,h_{t+1}^A,h_{t+1}^B)\nonumber\\
  =&\mu^k_t+\sum_{l\in L}q_{t+1}^l M_{a_t,b_t}^{kl}, \forall k\in K, \hbox{with $\mu_1^k=\mu^k$.} \label{eq: vector payoff player 2}
\end{align}
If $\mu^k=\mu^{*k}$ which takes the form as in (\ref{eq: initial regret mu}), then $\mu_t^k$ in the primal game can be seen as the difference between the expected realized payoff before stage $t$ and the expected total payoff using the security strategy if player 1 is of type $k$.

Player 2's security strategy at stage $t$ in $\tilde{\Gamma}_T^1(\mu,q)$ can be computed based on the backward recursive equation (\ref{eq: dynamic programming, T stage dual game 1}), and depends only on $t$, $\mu_t$ and $q_t$ \cite{sorin2002first}.
\begin{align}
  \tilde{V}^1_{n}(\mu_t,q_t)=&\min_{z\in \Delta(B)^{|L|}}\max_{a\in A} \sum_{b\in B} \bar{z}_{q_t,z}(b)\nonumber\\
   &\tilde{V}^1_{n-1}(\mu^+(\mu_t,a,b,z,q_t),q^+(b,z,q_t)), \label{eq: dynamic programming, T stage dual game 1}
\end{align}
where $n=T+1-t$.

Similarly, in type 2 dual game $\tilde{\Gamma}_T^2(p,\nu)$, player 1 also records two variables, the \emph{belief $p_t\in \Delta(K)$ on player 1's type} and the \emph{regret $\nu_t\in \mathbb{R}^{|L|}$} on player 2's type. The \emph{belief on player 1's type} is defined as
\begin{align}
  p_t^k=P(k|l,h_t^A,h_t^B), \forall k\in K, t=1,\ldots,T,
\end{align}
and is updated as below
\begin{align}
  p_{t+1}^k=&p^{+k}(a_t,r_t,p_t)=\frac{p_t^k r_t^k(a_t)}{\bar{r}_{p_t,r_t}(a_t)}, \forall k\in K, \label{eq: belief state player 1}
\end{align}
with $p_1=p$,
where $r_t^k=\sigma_t(k,h_t^A,h_t^B)\in \Delta(A)$, and $\bar{r}_{p_t,r_t}(a_t)=\sum_{k\in K}p_t^k r_t^k(a_t).$ The \emph{regret on player 2's type} is defined as
\begin{align*}
  \nu_t^l=\nu^l+\sum_{s=1}^{t-1}\mathbb{E}(M(k,l,a_s,b_s)|l,h_{s+1}^A,h_{s+1}^B),
\end{align*}
for all $l\in L, t=1,\ldots,T$, and is updated as below
\begin{align}
  \nu_{t+1}^l=&\nu^{+l}(\nu_t,a_t,b_t,r_t,p_t)\nonumber\\
  =&\nu_t^l+\mathbb{E}(M(k,l,a_t,b_t)|l,h_{t+1}^A,h_{t+1}^B)\nonumber\\
  =&\nu^l+\sum_{k\in K}p_{t+1}^k M_{a_t,b_t}^{kl},\forall l\in L, \hbox{with $\nu_1^l=\nu^l$}. \label{eq: vector payoff player 1}
\end{align}
If $\nu^l=\nu^{*l}$ which takes the form of (\ref{eq: initial regret nu}), then $\nu_t^l$ in the primal game can be seen as the difference between the expected realized payoff before stage $t$ and the expected total payoff using the security strategy if player 2 is of type $l$.

Player 1's security strategy at stage $t$ in $\tilde{\Gamma}_T^2(p,\nu)$ can be computed based on the backward recursive equation (\ref{eq: dynamic programming, T stage dual game 2}), and depends only on $t$, $p_t$ and $\nu_t$ \cite{sorin2002first}.
\begin{align}
  \tilde{V}_{n}^2(p_t,\nu_t)=& \max_{r\in \Delta(A)^{|K|}} \min_{b\in B} \sum_{a\in A} \bar{r}_{p_t,r}(a) \nonumber\\
  &\tilde{V}_{n-1}^2(p^+(a,r,p_t),\nu^+(\nu_t,a,b,r,p_t)),\label{eq: dynamic programming, T stage dual game 2}
\end{align}
where $n=T+1-t$.

From the analysis above, we see that the security strategies of player 1 and 2 in the corresponding dual games depend only on the fixed-sized sufficient statistics, $(t,p_t,\nu_t)$ and $(t,\mu_t,q_t)$, respectively, at stage $t$. Moreover, the sufficient statistics $(t,p_t,\nu_t)$ and $(t,\mu_t,q_t)$ are fully accessible to player 1 and 2, respectively, in the corresponding dual games. Based on the LP formulation of $V_T(p,q)$, we give the LP formulations to compute player 1's security strategy in type 2 dual game $\tilde{\Gamma}_T^2(p,\nu)$ and player 2's security strategy in type 1 dual game $\tilde{\Gamma}_T^1(\mu,q)$ as follows.
\begin{theorem}
\label{theorem: LP dual game T stage}
Consider type 2 dual game $\tilde{\Gamma}_T^2(p,\nu)$. Let $p_t$ and $\nu_t$ be the belief on player 1's type and the regret on player 2's type at stage $t$, respectively. The game value $\tilde{V}_n^2(p_t,\nu_t)$ of $n$ stage type 2 dual game $\tilde{\Gamma}_n^2(p_t,\nu_t)$ satisfies the following LP formulation, where $n=T+1-t$.
\begin{align}
  &\tilde{V}_n^2(p_t,\nu_t)=\max_{x\in X,u\in U,u_{:,0}\in\mathbb{R}^{|L|},\tilde{u}\in\mathbb{R}}\tilde{u} \label{eq: LP player 1 dual game}\\
  s.t.& u_{:,0}+\nu_t\geq \tilde{u}\mathbf{1} \label{eq: LP player 1-0 dual game}\\
  & \sum_{k\in K}{M^{kl}}^Tx_{k,h_1^A,h_1^B}+{u_{l,h_1^A,h_1^B}}^T\mathbf{1} \geq u_{l,0} \mathbf{1},\forall l\in L, \label{eq: LP player 1-1 dual game}\\
  & \sum_{k\in K}{M^{kl}}^Tx_{k,h_{t+1}^A,h_{t+1}^B}+{u_{l,h_{t+1}^A,h_{t+1}^B}}^T\mathbf{1} \geq u_{l,h_{t}^A,h_{t}^B}^{a_{t},b_{t}}\mathbf{1},\nonumber\\
   &\forall t=1,\ldots,n-1,l\in L,h_t^A\in H_t^A,h_t^B\in H_t^B,\label{eq: LP plyaer 1-2 dual game}
\end{align}
where $u_{l,h_n^A,h_n^B}$ is a zero matrix for all $l\in L$, $X$ is a set including all real vectors satisfying (\ref{eq: x constraint 1}-\ref{eq: x constraint 3}) with $x_{k,h_0^A,h_0^B}^{a_0}=p^k_t$, and $U$ is a real space of appropriate dimension. Player 1's security strategy $\tilde{\sigma}^*_t(k,p_t,\nu_t)$ at stage $t$ is
\begin{align}
  \tilde{\sigma}^{*}_t(k,p_t,\nu_t)=\frac{x_{k,h_1^A,h_1^B}^{*}}{p^k_t}. \label{eq: security strategy, player 1, dual game}
\end{align}

Similarly, for type 1 dual game $\tilde{\Gamma}_T^1(\mu,q)$, let $\mu_t$ and $q_t$ be the regret on player 1's type and the belief on player 2's type at stage $t$. The game value $\tilde{V}_n^1(\mu_t,q_t)$ of $n$ stage type 1 dual game $\tilde{\Gamma}_n^1(\mu_t,q_t)$ satisfies the following LP formulation, where $n=T+1-t$.
\begin{align}
  &\tilde{V}_n^1(\mu_t,q_t)=\min_{y\in Y,w\in W,w_{:,0}\in\mathbb{R}^{|K|},\tilde{w}\in\mathbb{R}} \tilde{w} \label{eq: LP player 2 dual game}\\
  s.t.& w_{:,0}+\mu_t \leq \tilde{w}\mathbf{1}, \label{eq: LP player 2-0 dual game}\\
  & \sum_{l\in L} M^{kl} y_{l,h_1^A,h_1^B}+w_{k,h_1^A,h_1^B}\mathbf{1} \leq w_{k,0}\mathbf{1},\forall k\in K, \label{eq: LP player 2-1 dual game}\\
  & \sum_{l\in L} M^{kl} y_{l,h_{t+1}^A,h_{t+1}^B}+w_{k,h_{t+1}^A,h_{t+1}^B}\mathbf{1} \leq w_{k,h_{t}^A,h_{t}^B}^{a_{t},b_{t}}\mathbf{1},\nonumber\\
  &\forall t=1,\ldots,n-1,k\in K,h_t^A\in H_t^A,h_t^B\in H_t^B, \label{eq: LP player 2-2 dual game}
\end{align}
where $w_{k,h_n^A,h_n^B}$ is a zero matrix for all $k\in K$, $Y$ is a set including all real vectors satisfying (\ref{eq: y constraint 1}-\ref{eq: y constraint 3}) with $y_{l,h_0^A,h_0^B}^{b_0}=q^l_t$, and $W$ is a real space of appropriate dimension. Player 2's security strategy $\tilde{\tau}^*_t(l,\mu_t,q_t)$ at stage $t$ is
\begin{align}
  \tilde{\tau}^*_t(l,\mu_t,q_t)=\frac{y_{l,h_1^A,h_1^B}^*}{q^l_t}. \label{eq: security strategy, player 2, dual game}
\end{align}
\end{theorem}
\begin{proof}
First, We have
\begin{align*}
  &\tilde{V}^2_n(p_t,\nu_t)\\
  =&\max_{x\in X}\min_{q\in \Delta(L)}\min_{\tau\in \mathcal{T}}\sum_{l\in L} q^l\left(\nu_t^l+\mathbb{E}(\sum_{s=1}^n M(k,l,a_s,b_s)|l)\right)\\
  =&\max_{x\in X}\min_{q\in \Delta(L)}\sum_{l\in L} q^l\left(\nu_t^l+\min_{\tau(l)\in \mathcal{T}(l)}\mathbb{E}(\sum_{s=1}^n M(k,l,a_s,b_s)|l)\right)\\
  =& \max_{x\in X}\min_{q\in \Delta(L)}\sum_{l\in L} q^l\left(\nu_t^l+u_{l,0}(x)\right).
\end{align*}

Define $\tilde{u}(x)=\min_{q\in \Delta(L)}\sum_{l\in L}q^l(\nu_t^l+u_{l,0}(x))$.
According to the dual theorem, given $x$, we have
\begin{align*}
  \tilde{u}(x)
  =&\max_{\tilde{u}\in \mathbb{R}}\tilde{u}\\
  s.t. & \nu_t+u_{:,0}(x) \geq \tilde{u}\mathbf{1},
\end{align*}
where $u_{:,0}(x)$ satisfies (\ref{eq: u(x)}-\ref{eq: u(x) 2}) with the horizon to be $n$. Therefore, following the same steps in the proof of Theorem \ref{theorem: LP primal game T stage} to show $\hat{u}=u$, we have
\begin{align*}
  \tilde{u}(x)=&\max_{u\in U,u_{:,0}\in\mathbb{R}^{|L|},\tilde{u}\in \mathbb{R}}\tilde{u}\\
  s.t. & \nu_t+u_{:,0} \geq \tilde{u}\mathbf{1},\\
  & \sum_{k\in K}{M^{kl}}^Tx_{k,h_1^A,h_1^B}+{u_{l,h_1^A,h_1^B}}^T\mathbf{1} \geq u_{l,0} \mathbf{1},\forall l\in L, \\
  & \sum_{k\in K}{M^{kl}}^Tx_{k,h_{t+1}^A,h_{t+1}^B}+{u_{l,h_{t+1}^A,h_{t+1}^B}}^T\mathbf{1} \geq u_{l,h_{t}^A,h_{t}^B}^{a_{t},b_{t}}\mathbf{1}, \\
  &\forall t=1,\ldots,n-1,l\in L,h_t^A\in H_t^A,h_t^B\in H_t^B.
\end{align*}
Hence, equation (\ref{eq: LP player 1 dual game}-\ref{eq: LP plyaer 1-2 dual game}) is shown. Player 1's security strategy at stage $t$ in dual game $\tilde{\Gamma}_T^2(p,\nu)$ can be seen as player 1's security strategy at stage $1$ in dual game $\tilde{\Gamma}_n^2(p_t,\nu_t)$. Hence, according to equation (\ref{eq: x}), we have $\tilde{\sigma}_t^*(k,p_t,\nu_t)=x_{k,h_1^A,h_1^B}^*/p_t^k$.

Following the same steps, we show equation (\ref{eq: LP player 2 dual game}-\ref{eq: LP player 2-2 dual game}) is also true, and player 2's security strategy at stage $t$ is as in (\ref{eq: security strategy, player 2,   dual game}).
\end{proof}

Now, let us get back to the primal $T$-stage repeated Bayesian game $\Gamma_T(p,q)$. It was shown in \cite{rosenberg1998duality,sorin2002first,de1996repeated} that if $\nu^*$ is the optimal solution to $\max_{\nu\in\mathbb{R}^{|L|}}\{\tilde{V}_T^2(q,\nu)-q^T\nu\}$, then player 1's security strategy in type 2 dual game $\tilde{\Gamma}^2_T(p,\nu^*)$ is also the player's security strategy in the primal game $\Gamma_T(p,q)$, and that if $\mu^*$ is the optimal solution to $\min_{\mu\in\mathbb{R}^{|K|}}\{\tilde{V}_T^1(\mu,q)-p^T\mu\}$, then player s's security strategy in type 1 dual game $\tilde{\Gamma}^1_T(\mu^*,q)$ is also the player's security strategy in the primal game $\Gamma_T(p,q)$. Since Lemma \ref{lemma: optimal solution T stage} shows that $\nu^*$ and $\mu^*$ are the regrets on player 2 and 1's type, respectively, we have the following corollary.
\begin{corollary}
\label{theorem: security strategy relation between primal and dual games}
Consider a $T$-stage repeated Bayesian game $\Gamma_T(p,q)$ and its dual games $\tilde{\Gamma}_T^1(\mu,q)$ and $\tilde{\Gamma}_T^2(p,\nu)$. Player 1's security strategy $\tilde{\sigma}^*\in \Sigma$, which depends only on $t$, $p_t$ and $\nu_t$ at stage $t$, in type 2 dual game $\tilde{\Gamma}^2_T(p,\nu^*)$ is also player 1's security strategy in the primal game $\Gamma_T(p,q)$, where $\nu^*$ is given in (\ref{eq: initial regret nu}).

Similarly, player 2's security strategy $\tilde{\tau}^*\in \mathcal{T}$, which depends only on $t$, $\mu_t$ and $q_t$ at stage $t$, in type 1 dual game $\tilde{\Gamma}_T^1(\mu^*,q)$ is also player 2's security strategy in the primal game $\Gamma_T(p,q)$, where $\mu^*$ is given in (\ref{eq: initial regret mu}).
\end{corollary}

According to Corollary \ref{theorem: security strategy relation between primal and dual games}, we can compute player 1's security strategy in the following way. First, compute the initial regret, $\nu^*$, on player 2's type. Stage by stage, update $p_t$ and $\nu_t$, and compute the security strategy based on $p_t$, $\nu_t$ and $t$ in the dual game. Player 2's security strategy is computed in the same way.
\begin{algorithm}{Player 1's security strategy based on fixed-sized sufficient statistics}\hfill{}\\
\label{algorithm: player 1's strategy T stage}
 \begin{enumerate}
   \item Initialization
        \begin{itemize}
          \item Compute $u_{:,0}^*$ based on LP (\ref{eq: LP player   1}-\ref{eq: LP plyaer 1-2}).
          \item Set $t=1$, $p_t=p$, and $\nu_t=-u_{:,0}^*$.
        \end{itemize}
   \item Compute player 1's security strategy $\tilde{\sigma}_t^*$ at stage $t$ according to (\ref{eq: security strategy,   player 1, dual game}) based on LP (\ref{eq: LP player 1 dual game}-\ref{eq: LP plyaer 1-2 dual game}).
   \item Choose an action in $A$ according to $\tilde{\sigma}_t^*$, and announce the action publicly. Meanwhile, read player 2's current action.
   \item If $t=T$, then go to step 6. Otherwise, update $p_{t+1}$ and $\nu_{t+1}$ according to (\ref{eq: belief state player 1}) and (\ref{eq: vector payoff player   1}), respectively.
   \item Update $t=t+1$ and go to step 2.
   \item End.
 \end{enumerate}
\end{algorithm}
\begin{algorithm}{Player 2's security strategy based on fixed-sized sufficient statistics}\hfill{}\\
\label{algorithm: player 2's algorithm T stage}
  \begin{enumerate}
    \item Initialization
        \begin{itemize}
          \item Compute $w_{:,0}^*$ based on LP (\ref{eq: LP player   2}-\ref{eq: LP player 2 2}).
          \item Set $t=1$, $\mu_t=-w_{:,0}^*$, and $q_t=q$.
        \end{itemize}
    \item Compute Player 2's security strategy $\tilde{\tau}_t^*$ at stage $t$ according to (\ref{eq: security strategy, player 2,   dual game}) based on LP (\ref{eq: LP player 2 dual game}-\ref{eq: LP player 2-2 dual game}).
    \item Choose an action in $B$ according to $\tilde{\tau}_t^*$, and announce it publicly. Meanwhile, read player 1's current action.
    \item If $t=T$, then go to step 6. Otherwise, update $q_{t+1}$ and $\mu_{t+1}$ according to (\ref{eq: belief state player 2}) and (\ref{eq: vector payoff   player 2}), respectively.
    \item Update $t=t+1$ and go to step 2.
    \item End.
  \end{enumerate}
\end{algorithm}

\section{$\lambda$-Discounted Repeated Bayesian Games}
A two-player zero-sum $\lambda$-discounted repeated Bayesian game, which is simply called discounted game or discounted primal game in the rest of this paper,  is specified by the same seven-tuple $(K,L,A,B,M,p_0,q_0)$ and played in the same way as in a two-player zero-sum $T$-stage repeated Bayesian game. The payoff of player 1 at stage $t$ is $\lambda(1-\lambda)^{t-1}M(k,l,a_t,b_t)$, where $\lambda\in (0,1)$, and the game is played for infinite horizon. Correspondingly, the strategy spaces $\Sigma$ and $\mathcal{T}$ are defined for infinite horizon. The total payoff of the discounted game with initial probability $p_0,q_0$ and strategies $\sigma$ and $\tau$ is defined as
\begin{align*}
  \gamma_{\lambda}(p_0,q_0,\sigma,\tau)=\mathbb{E}_{p_0,q_0,\sigma,\tau}\left(\sum_{t=1}^\infty \lambda(1-\lambda)^{t-1} M(k,l,a_t,b_t)\right).
\end{align*}

The discounted game $\Gamma_\lambda(p_0,q_0)$ is defined as a two-player zero-sum repeated Bayesian game equipped with initial distribution $p_0$ and $q_0$, strategy spaces $\Sigma$ and $\mathcal{T}$, and payoff function $\gamma_\lambda(p_0,q_0,\sigma,\tau)$. The security strategies $\sigma^*$ and $\tau^*$, and security levels $\underline{V}_\lambda(p_0,q_0)$ and $\overline{V}_\lambda(p_0,q_0)$ are defined in the same way as in a $T$-stage repeated Bayesian game. Since $\gamma_\lambda(p_0,q_0,\sigma,\tau)$ is bilinear over $\sigma$ and $\tau$, the discounted game $\Gamma_\lambda(p_0,q_0)$ has a value $V_\lambda(p_0,q_0)$ according to Sion's minmax Theorem \cite{sion1958general}, i.e. $V_\lambda(p_0,q_0)=\underline{V}_\lambda(p_0,q_0)=\overline{V}_\lambda(p_0,q_0)$.

\subsection{Dual games, security strategies, and sufficient statistics}

A discounted game is played for infinite stages, and the history action space is infinite, too. It is not practical to design behavior strategies that directly depends on history actions, and it is necessary to find a sufficient statistics for decision making. A candidate sufficient statistics in the discounted game $\Gamma_\lambda(p,q)$ is the belief state pair $(p_t,q_t)$ \cite{sorin2002first,rosenberg1998duality}. The belief state pair is, unfortunately, not fully available to either player after the first stage, since $(p_t,q_t)$ depends on both players' strategies according to (\ref{eq: belief state player 1}) and (\ref{eq: belief state player 2}). The objective in this section is to find, for every player, the fully available sufficient statistics and the corresponding security strategy that depends on the sufficient statistics. We will use the same technique as that in the $T$-stage game to find the fully available sufficient statistics and the corresponding security strategies. Let us start from the dual games of the discounted game.

A discounted game $\Gamma_\lambda(p,q)$ also has two dual games. The discounted type 1 dual game $\tilde{\Gamma}_\lambda^1(\mu,q)$ is with respect to the first parameter $p$, where $\mu\in \mathbb{R}^{|K|}$ is the initial regret with respect to player 1's type. The discounted type 1 dual game $\tilde{\Gamma}_\lambda^1(\mu,q)$ is played the same as in the $T$-stage type 1 dual game $\tilde{\Gamma}_T^1(\mu,q)$, except that the discounted game is played for infinite horizon. Let $p$ be player 1's strategy to choose his own type. Player 1's payoff is
\begin{align*}
&\tilde{\gamma}^1_\lambda(\mu,q,p,\sigma,\tau)\\
=&\mathbb{E}_{p,q,\sigma,\tau}\left(\mu^k+\sum_{t=1}^\infty \lambda(1-\lambda)^{t-1}M(k,l,a_t,b_t)\right).
\end{align*}
The discounted type 2 dual game $\tilde{\Gamma}_\lambda^2(p,\nu)$ is defined with respect to the second parameter $q$, where $\nu\in \mathbb{R}^{|L|}$ is the initial regret with respect to player 2's type. The discounted type 2 dual game $\tilde{\Gamma}_\lambda^2(p,\nu)$ is played the same as in the $T$-stage type 2 dual game $\tilde{\Gamma}_T^2(p,\nu)$, except that the discounted game is played for infinite horizon. Let $q$ be player 2's strategy to choose his type. Player 1's payoff is
\begin{align*}
&\tilde{\gamma}^2_\lambda(p,\nu,q,\sigma,\tau)\\
=&\mathbb{E}_{p,q,\sigma,\tau}\left(\nu^l+\sum_{t=1}^\infty \lambda(1-\lambda)^{t-1}M(k,l,a_t,b_t)\right).
\end{align*}

Both dual games, $\tilde{\Gamma}_\lambda^1(\mu,q)$ and $\tilde{\Gamma}_\lambda^2(p,\nu)$, have values denoted by $\tilde{V}^1_\lambda(\mu,q)$ and $\tilde{V}^2_\lambda(p,\nu)$, respectively. The game values of the discounted dual games are related to the game value of the discounted primal game in the following way \cite{sorin2002first}.
\begin{align}
  \tilde{V}^1_\lambda(\mu,q)=&\max_{p\in \Delta(K)}\{V_\lambda(p,q)+p^T\mu\} \label{eq: game value relations 1, discounted game}\\
  V_\lambda(p,q)=& \min_{\mu\in\mathbb{R}^{|K|}} \{\tilde{V}^1_\lambda(\mu,p)-p^T\mu\} \label{eq: game value relations 2, discounted game} \\
  \tilde{V}^2_\lambda(p,\nu)=& \min_{q\in \Delta(L)} \{ V_\lambda(p,q)+q^T\nu\} \label{eq: game value relations 3, discounted game}\\
  V_\lambda(p,q)=& \max_{\nu\in\mathbb{R}^{|L|}} \{ \tilde{V}^2_\lambda(p,\nu)-q^T\nu\} \label{eq: game value relations 4, discounted game}
\end{align}

Let $\mu^*$ and $\nu^*$ be the optimal solution to the optimal problem on the right hand side of (\ref{eq: game value   relations 2, discounted game}) and (\ref{eq: game value   relations 4, discounted game}), respectively. Player 2's security strategy in discounted type 1 dual game $\tilde{\Gamma}_\lambda^1(\mu^*,q)$ is also his security strategy in the discounted primal game $\Gamma_\lambda(p,q)$ \cite{sorin2002first,de1996repeated,rosenberg1998duality}. Player 1's security strategy in discounted type 2 dual game $\tilde{\Gamma}_\lambda^2(p,\nu^*)$ is his security strategy in the discounted primal game $\Gamma_\lambda(p,q)$ \cite{sorin2002first}. The following lemma tells us what $\mu^*$ and $\nu^*$ are. The proof is the same as the proof of Lemma \ref{lemma: optimal solution T stage}
\begin{lemma}
  \label{lemma: optimal solution discounted game}
Consider a discounted game $\Gamma_\lambda(p,q)$. Let $\sigma^*_{p,q}$ and $\tau^*_{p,q}$ be player 1 and 2's security strategies, and $x^*_{p,q}$ and $y_{p,q}^*$ be the corresponding optimal realization plans of player 1 and 2. Define
\begin{align*}
  u_{l,0;\lambda}(x)=&\min_{\tau(l)\in \mathcal{T}(l)} \mathbb{E}\left(\sum_{t=1}^\infty \lambda(1-\lambda)^{t-1} M(k,l,a_t,b_t)|l\right),\\
  w_{k,0;\lambda}(y)=&\max_{\sigma(k)\in \Sigma(k)}\mathbb{E}\left(\sum_{t=1}^\infty \lambda(1-\lambda)^{t-1} M(k,l,a_t,b_t)|k\right),
\end{align*}
where $x$ and $y$ are player 1 and 2's realization plans.
The optimal solution $\mu^*$ to the optimal problem $\min_{\mu_\in\mathbb{R}^{|K|}}\{\tilde{V}^1_\lambda(\mu,q)-p^T\mu\}$ is
\begin{align}
\mu^{*k}=-w_{k,0;\lambda}(y^*_{p,q}), \forall k\in K. \label{eq: initial regret player 1 discounted}
\end{align}

The optimal solution $\nu^*$ to the optimal problem $\max_{\nu\in\mathbb{R}^{|L|}} \{ \tilde{V}^2_\lambda(p,\nu)-q^T\nu\}$ is
\begin{align}
\nu^{*l}=-u_{l,0;\lambda}(x^*_{p,q}), \forall l\in L. \label{eq: initial regret player 2 discounted}
\end{align}
\end{lemma}

Now that we've found the special initial regrets in the dual games, our next step is to study players' security strategies in the dual games. With a little abuse of notation, in discounted games, $\mu_t$ and $\nu_t$ are called the \emph{anti-discounted regret} on player 1 and 2's types, respectively. In discounted type 1 dual game $\tilde{\Gamma}_\lambda^1(\mu,q)$, the anti-discounted regret $\mu_t$ on player 1's type is defined as
\begin{align*}
  \mu_t^k=\frac{\mu^k+\sum_{s=1}^{t-1}\lambda(1-\lambda)^{s-1} \mathbb{E}(M(k,l,a_s,b_s)|k,h_{s+1}^A,h_{s+1}^B)}{(1-\lambda)^{t-1}},
\end{align*}
for all $k\in K, t=1,2,\ldots$, and is updated as follows
\begin{align}
  \mu_{t+1}^k=\mu^+(\mu_t,a_t,b_t,z_t,q_t)=\frac{\mu_t^k+\lambda \sum_{l\in L} q_{t+1}^lM^{kl}_{a_t,b_t}}{1-\lambda} ,  \label{eq: mu discounted game}
\end{align}
for all $k\in K,$ with $\mu_1=\mu$,
where $q_t$ is the belief on player 2's type defined as in (\ref{eq: belief state player 2}), and updated as in (\ref{eq: belief state player 2}).

In discounted type 2 dual game $\tilde{\Gamma}_\lambda^2(p,\nu)$, the \emph{anti-discounted regret $\nu_t$ on player 2' type} is defined as
\begin{align*}
  \nu_t^l=\frac{\nu^l+\sum_{s=1}^{t-1}\lambda(1-\lambda)^{s-1} \mathbb{E}(M(k,l,a_s,b_s)|l,h_{s+1}^A,h_{s+1}^B)}{(1-\lambda)^{t-1}},
\end{align*}
for all $l\in L, t=1,2,\ldots$, and is updated as follows
\begin{align}
  \nu_{t+1}^l=&\nu^+(\nu_t,a_t,b_t,r_t,p_t)=\frac{\nu_t^l+\lambda \sum_{k\in K} p_{t+1}^k M^{kl}_{a_t,b_t}} {1-\lambda}, \label{eq: nu discounted game}
\end{align}
for all $l\in L$, with $\nu_1=\nu$,
where $p_{t}$ is the belief on player 1's type defined as in (\ref{eq: belief state player 1}), and updated as in (\ref{eq: belief state player 1}).

Player 2's security strategy in discounted type 1 dual game can be found by solving equation (\ref{eq: dynamic pragraming,   discounted dual game 1}), and depends only on $\mu_t$ and $q_t$ at stage $t$. Player 1's security strategy in discounted type 2 dual game can be computed by solving equation (\ref{eq: dynamic pragraming,   discounted dual game 2}), and depends only on $q_t$ and $\nu_t$ at stage $t$. The desirable property of $(\mu_t,q_t)$ and $(p_t,\nu_t)$ is that they are fully available to player 2 and $1$, respectively.
\begin{align}
  \tilde{V}^1_\lambda(\mu,q)=&\min_{z\in\Delta(B)^{|L|}}\max_{a\in A}(1-\lambda)\sum_{b\in B}\bar{z}_{q,z}(b)\nonumber \\
  &\tilde{V}^1_\lambda(\mu^+(\mu,a,b,z,q),q^+(b,z,q)).\label{eq: dynamic pragraming, discounted dual game 1} \\
  \tilde{V}^2_\lambda(p,\nu)=&\max_{r\in\Delta(A)^{|K|}} \min_{b\in B} (1-\lambda) \sum_{a\in A} \bar{r}_{p,r}(a)\nonumber\\
   &\tilde{V}^2_\lambda(p^+(a,r,p),\nu^+(\nu,a,b,r,p)).\label{eq: dynamic pragraming, discounted dual game 2}
\end{align}

Finally, we are ready to investigate players' sufficient statistics and the corresponding security strategies in the discounted primal game $\Gamma_\lambda(p,q)$.
\begin{corollary}{\cite{sorin2002first}}
  \label{theorem: security strategy and sufficient statitics in discounted games}
Consider a discounted game $\Gamma_\lambda(p,q)$ and its dual games $\tilde{\Gamma}_\lambda^1(\mu,q)$ and $\tilde{\Gamma}_\lambda^2(p,\nu)$. Let $\nu^*$ take the form of (\ref{eq: initial regret player 2 discounted}). Player 1's security strategy $\tilde{\sigma}^*$, which depends only on $(p_t,\nu_t)$ at stage $t$, in discounted type 2 dual game $\tilde{\Gamma}_\lambda^2(p,\nu^*)$ is also player 1's security strategy in the discounted primal game $\Gamma_\lambda(p,q)$.

Similarly, let $\mu^*$ take the form of (\ref{eq: initial regret player 1 discounted}). Player 2's security strategy $\tilde{\tau}^*$, which depends only on $(\mu_t,q_t)$ at stage $t$, in discounted type 1 dual game $\tilde{\Gamma}_\lambda^1(\mu^*,q)$ is also player 2's security strategy in the discounted primal game $\Gamma_\lambda(p,q)$.
\end{corollary}

\subsection{Approximating the initial regret states $\mu^*$ and $\nu^*$}
To compute players' security strategies in the discounted primal game, the first thing is to compute the initial regrets, $\mu^*=-w_{:,0;\lambda}(y^*)$ and $\nu^*=-u_{:,0;\lambda}(x^*)$, which is a non-convex problem in the variables \cite{sandholm2010state}. Therefore, we consider using the game value of a $\lambda$-discounted $T$-stage game to approximate the game value of the discounted game with infinite horizon, and further find approximated $\mu^*$ and $\nu^*$.

A $\lambda$-discounted $T$-stage repeated Bayesian game $\Gamma_{\lambda,T}(p,q)$ is a truncated version of the $\lambda$-discounted repeated Bayesian game $\Gamma_\lambda(p,q)$ in which the time horizon is $T$. We denote the payoff and the game value of the truncated discounted game as $\gamma_{\lambda,T}(p,q,\sigma,\tau)$ and $V_{\lambda,T}(p,q)$, respectively. In game $\Gamma_{\lambda,T}(p,q)$, we define anti-discounted weighted future security payoffs $u_{l,h_t^A,h_t^B;\lambda,T}^{a_t,b_t}$ and $w_{k,h_t^A,h_t^B;\lambda,T}^{a_t,b_t}$ for $t=0,\ldots,T-1$ as follows
\begin{align*}
  &u_{l,h_t^A,h_t^B;\lambda,T}^{a_t,b_t}(x)=(1-\lambda)^{-t}\min_{\tau_{t+1:T}(l)\in \mathcal{T}_{t+1:T}(l)} \sum_{k\in K} x_{k,h_t^A,h_t^B}^{a_t} \\
  &\mathbb{E}\left(\sum_{s=t+1}^T \lambda(1-\lambda)^{s-1}M(k,l,a_s,b_s)|k,l,h_{t+1}^A,h_{t+1}^B\right),\\
  &w_{k,h_t^A,h_t^B;\lambda,T}^{a_t,b_t}(y)=(1-\lambda)^{-t}\max_{\sigma_{t+1:T}(k)\in \Sigma_{t+1:T}(k)} \sum_{l\in L} y_{l,h_t^A,h_t^B}^{b_t} \\
  &\mathbb{E}\left(\sum_{s=t+1}^T \lambda(1-\lambda)^{s-1} M(k,l,a_s,b_s)|k,l,h_{t+1}^A,h_{t+1}^B\right),
\end{align*}
where $h_{t+1}^A=(h_t^A,a_t)$ and $h_{t+1}^B=(h_t^B,b_t)$, and the pairs here indicate concatenation.

Notice that $u_{l,h_0^A,h_0^B;\lambda,T}^{a_0,b_0}(x)$ and $w_{k,h_0^A,h_0^B;\lambda,T}^{a_0,b_0}(y)$, also denoted as $u_{l,0;\lambda,T}(x)$ and $w_{k,0;\lambda,T}(y)$, are truncated versions of $u_{l,0;\lambda}(x)$ and $w_{k,0;\lambda}(y)$, respectively. We can use $u_{l,0;\lambda,T}(x^\star)$ and $w_{k,0;\lambda,T}(y^\star)$ to approximate $u_{l,0;\lambda}(x^*)$ and $w_{k,0;\lambda}(y^*)$, and hence $\nu^*$ and $\mu^*$. Here, $x^*$ and $y^*$ are player 1 and 2's optimal realization plan in game $\Gamma_\lambda(p,q)$, and $x^\star$ and $y^\star$ are player 1 and 2's security strategies in game $\Gamma_{\lambda,T}(p,q)$. The following theorem provides linear programs to compute $u_{l,0;\lambda,T}(x^\star)$ and $w_{k,0;\lambda,T}(y^\star)$.
\begin{theorem}
\label{theorem: LP for parameter approximation}
Consider a $\lambda$-discounted $T$-stage repeated Bayesian game $\Gamma_{\lambda,T}(p,q)$. Its game value $V_{\lambda,T}(p,q)$ satisfies
\begin{align}
&V_{\lambda,T}(p,q)= \max_{x\in X,u_{\lambda,T}\in U,u_{:,0;\lambda,T}\in\mathbb{R}^{|L|}}\sum q^l u_{l,0;\lambda,T} \label{eq: LP discounted game 1} \\
s.t. & \lambda\sum_{k\in K}{M^{kl}}^T x_{k,h_1^A,h_1^B}+(1-\lambda){u_{l,h_1^A,h_1^B;\lambda,T}}^T\mathbf{1}\nonumber\\
 &\geq u_{l,0;\lambda,T}\mathbf{1}, \forall l\in L,\label{eq: LP discounted game 1-1}\\
& \lambda\sum_{k\in K}{M^{kl}}^T x_{k,h_{t+1}^A,h_{t+1}^B}+(1-\lambda){u_{l,h_{t+1}^A,h_{t+1}^B;\lambda,T}}^T\mathbf{1} \nonumber\\
&\geq u_{l,h_t^A,h_t^B;\lambda,T}^{a_t,b_t}\mathbf{1}, \forall t=1,\ldots,T-1, l\in L,\nonumber\\
& \forall h_{t+1}\in H_{t+1}^A,h_{t+1}^B\in H_{t+1}^B, \label{eq: LP discounted game 1-2}
\end{align}
where $h_{t+1}^A=(h_t^A,a_t)$, $h_{t+1}^B=(h_t^B,b_t)$, $X$ is a set including all real vectors satisfying (\ref{eq: x constraint 1}-\ref{eq: x constraint 3}), and $U$ is a real space of an appropriate dimension. The optimal solution $x^\star_{\lambda,T}$ is player 1's optimal realization plan in game $\Gamma_{\lambda,T}(p,q)$, and its anti-discounted weighted future security payoff at stage $0$ is the optimal solution $u^\star_{:,0;\lambda,T}$, i.e. $u_{:,0;\lambda,T}(x^\star)=u^\star_{:,0;\lambda,T}$.

Dually, $V_{\lambda,T}(p,q)$ also satisfies
\begin{align}
&V_{\lambda,T}(p,q)=\min_{y\in Y,w_{\lambda,T}\in W,w_{:,0;\lambda,T}\in\mathbb{R}^{|K|}}\sum_{k\in K} p^k w_{k,0;\lambda,T} \label{eq: LP discounted game player 2}\\
  s.t. & \lambda\sum_{l\in L} M^{kl} y_{l,h_1^A,h_1^B}+(1-\lambda)w_{k,h_1^A,h_1^B;\lambda,T}\mathbf{1}\nonumber\\
   &\leq w_{k,0;\lambda,T}\mathbf{1},\forall k\in K, \label{eq: LP discounted game 2-1}\\
  & \lambda \sum_{l\in L} M^{kl} y_{l,h_{t+1}^A,h_{t+1}^B}+ (1-\lambda) w_{k,h_{t+1}^A,h_{t+1}^B;\lambda,T}\mathbf{1} \nonumber \\
  &\leq w_{k,h_{t}^A,h_{t}^B;\lambda,T}^{a_{t},b_{t}}\mathbf{1},\forall t=1,\ldots,T-1,k\in K,\nonumber\\
&\forall h_{t+1}^A\in H_{t+1}^A,h_{t+1}^B\in H_{t+1}^B, \label{eq: LP discounted game 2-2}
\end{align}
where $h_{t+1}^A=(h_t^A,a_t)$, $h_{t+1}^B=(h_t^B,b_t)$, $Y$ is a set including all real vectors satisfying (\ref{eq: y constraint 1}-\ref{eq: y constraint 3}), and $W$ is a real space of an appropriate dimension. The optimal solution $y^\star$ is player 2's optimal realization plan in game $\Gamma_{\lambda,T}(p,q)$, and its anti-discounted weighted future security payoff at stage $0$ is the optimal solution $w^\star_{:,0;\lambda,T}$, i.e. $w_{:,0;\lambda,T}(y^\star)=w^\star_{:,0;\lambda,T}$.
\end{theorem}
\begin{proof}
Following the same steps as in the proof of Lemma \ref{lemma: weighted future security payoff}, we have for all $t=0,\ldots,T-1$,
\begin{align}
 &u_{l,h_t^A,h_t^B;\lambda,T}^{a_t,b_t}(x)= \min_{\tau_{t+1}(l,h_{t+1}^A,h_{t+1}^B)\in \Delta(B)}\left(\lambda \sum_{k\in K}x_{k,h_{t+1}^A,h_{t+1}^B}^T\right.\nonumber\\
 &\left.M^{kl}+(1-\lambda)\mathbf{1}^Tu_{l,h_{t+1}^A,h_{t+1}^B;\lambda,T}(x)\right) \tau_{t+1}(l,h_{t+1}^A,h_{t+1}^B),\label{eq: u recursive discount} \\
 &w_{k,h_t^A,h_t^B;\lambda,T}^{a_t,b_t}(y)= \max_{\sigma_{t+1}(k,h_{t+1}^A,h_{t+1}^B)\in \Delta(A)}\sigma_{t+1}(k,h_{t+1}^A,h_{t+1}^B)^T\nonumber\\
 &\left(\lambda \sum_{l\in L}M^{kl}y_{l,h_{t+1}^A,h_{t+1}^B}+(1-\lambda )w_{k,h_{t+1}^A,h_{t+1}^B;\lambda,T}(y)\mathbf{1}\right)\label{eq: w recursive discount}
\end{align}

Following the same steps as in the proof of Theorem \ref{theorem: LP primal game T stage}, we show Theorem \ref{theorem: LP for parameter approximation} is true.
\end{proof}

With $u^\star(:,0;\lambda,T)$ and $w^\star(:,0;\lambda,T)$ computed based on (\ref{eq: LP discounted game 1}-\ref{eq: LP discounted game 1-2}) and (\ref{eq: LP discounted game player 2}-\ref{eq: LP discounted game 2-2}), according to Lemma \ref{lemma: optimal solution discounted game}, we approximate $\mu^*$ and $\nu^*$ as
\begin{align}
  \mu^\dag=&-w^\star(:,0;\lambda,T), \hbox{and} \label{eq: approximated mu}\\
  \nu^\dag=&-u^\star(:,0;\lambda,T), \label{eq: approximated nu}
\end{align}
respectively.

\subsection{Approximating the security strategies $\tilde{\sigma}^*$ and $\tilde{\tau}^*$ in dual games}
Now that we have constructed an LP to compute the approximated initial regrets $\mu^\dag$ and $\nu^\dag$ in the dual games, the next step is to compute the security strategy in a discounted dual game, which will serve as the security strategy of the corresponding player in the discounted primal game.

Computing the security strategies and the game values in dual games $\tilde{\Gamma}_\lambda^1(\mu^*,q)$ and $\tilde{\Gamma}_\lambda^2(p,\nu^*)$ is non-convex \cite{sandholm2010state}. Therefore, we use the game values of truncated discounted dual games to approximate the game value of the discounted dual games, and then compute approximated security strategies based on the approximated game value.


A $\lambda$-discounted $T$-stage type 1 dual game $\tilde{\Gamma}^1_{\lambda,T}(\mu,q)$ is a truncated discounted type 1 dual game $\tilde{V}^1_\lambda(\mu,q)$ with time horizon to be $T$ stages. Following the same step as in the proof of Proposition 4.22 in \cite{sorin2002first}, the game value $\tilde{V}_{\lambda,T+1}^1(\mu,q)$ of the $\lambda$-discounted $T+1$-stage type 1 dual game $\tilde{\Gamma}^1_{\lambda,T+1}(\mu,q)$ satisfies
\begin{align}
  &\tilde{V}_{\lambda,T+1}^1(\mu,q)=\min_{z\in\Delta(B)^{|L|}}\max_{a\in A} (1-\lambda)\sum_{b\in B} \bar{z}_{q,z}(b)\nonumber \\ &\tilde{V}_{\lambda,T}^1(\mu^+(\mu,a,b,z,q),q^+(b,z,q)), \label{eq: dynamic programming, truncated discounted dual game 1}
\end{align}
with $\tilde{V}_{\lambda,0}^1(\mu,q)=\max \{\mu\}$.
Moreover, since $\tilde{\Gamma}^1_{\lambda,T}(\mu,q)$ is a dual game of $\Gamma_{\lambda,T}(p,q)$ with respect to the first parameter $p$, their game values satisfy
\begin{align}
  \tilde{V}_{\lambda,T}^1(\mu,q)=&\max_{p\in \Delta(K)} \{V_{\lambda,T}(p,q)+p^T\mu\}, \label{eq: game value relation 1, truncated discounted game}\\
  V_{\lambda,T}(p,q)=&\min_{\mu\in\mathbb{R}^{|K|}} \{\tilde{V}^1_{\lambda,T}(\mu,q)-p^T\mu\}.\label{eq: game value relation 2, truncated discounted game}
\end{align}

Similarly, a type 2 $\lambda$-discounted $T$-stage dual game $\tilde{\Gamma}^2_{\lambda,T}(p,\nu)$ is the truncated version of discounted type 2 dual game with time horizon $T$. The game value $\tilde{V}^2_{\lambda,T+1}(p,\nu)$ of the truncated discounted type 2 dual game $\tilde{\Gamma}^2_{\lambda,T+1}(p,\nu)$ satisfies
\begin{align}
  &\tilde{V}_{\lambda,T+1}^2(p,\nu)=\max_{r\in\Delta(A)^{|K|}}\min_{b\in B}(1-\lambda) \sum_{a\in A} \bar{r}_{p,r}(a)\nonumber\\
  &\tilde{V}_{\lambda,T}^2(p^+(a,r,p),\nu^+(\nu,a,b,r,p)), \label{eq: dynamic programming, truncated discounted dual game 2}
\end{align}
with $\tilde{V}_{\lambda,0}(p,\nu)=\min\{\nu\}$. The truncated discounted type 2 dual game $\tilde{\Gamma}^2_{\lambda,T}(p,\nu)$ is the dual game of the truncated discounted game $\Gamma_{\lambda,T}(p,q)$ with respect to the second parameter $q$, and hence their game values satisfy
\begin{align}
  \tilde{V}_{\lambda,T}^2(p,\nu)=&\min_{q\in\Delta(L)} \{V_{\lambda,T}(p,q)+q^T\nu\},\label{eq: game value relation 3, truncated discounted game}\\
  V_{\lambda,T}(p,q)=&\max_{\nu\in\mathbb{R}^{|L|}} \{\tilde{V}_{\lambda,T}^2(p,\nu)-q^T\nu\}.\label{eq: game value relation 4, truncated discounted game}
\end{align}

Based on the relations between the game values of the discounted game, the truncated discounted game $\Gamma_{\lambda,T}(p,q)$, and their dual games, we have the following lemma.
\begin{lemma}
  \label{lemma: supreme norm equality}
Consider a discounted game $\Gamma_\lambda(p,q)$ and its dual games $\tilde{\Gamma}_\lambda^1(\mu,q)$ and $\tilde{\Gamma}_\lambda^2(p,\nu)$, and a $T$-stage discounted game $\Gamma_{\lambda,T}(p,q)$ and its dual games $\tilde{\Gamma}_{\lambda,T}^1(\mu,q)$ and $\tilde{\Gamma}_{\lambda,T}^2(p,\nu)$. Their game values satisfy
\begin{align}
  \|V_\lambda-V_{\lambda,T}\|_{\sup}=\|\tilde{V}_\lambda^1-\tilde{V}_{\lambda,T}^1\|_{\sup}=\|\tilde{V}_\lambda^2-\tilde{V}_{\lambda,T}^2\|_{\sup} \label{eq: game value equalities}
\end{align}
\end{lemma}
\begin{proof}
First, we show that $\|V_\lambda-V_{\lambda,T}\|_{\sup}\leq\|\tilde{V}_\lambda^1-\tilde{V}_{\lambda,T}^1\|_{\sup}$. According to equation (\ref{eq: game value   relations 2, discounted game}) and (\ref{eq:   game value relation 2, truncated discounted game}), we have
\begin{align*}
  &|V_\lambda(p,q)-V_{\lambda,T}(p,q)|\\
  =&|\min_{\mu\in\mathbb{R}}\{\tilde{V}_{\lambda}^1(\mu,q)-p^T\mu\} -\min_{\mu\in\mathbb{R}}\{\tilde{V}_{\lambda,T}^1(\mu,q)-p^T\mu\}|.
\end{align*}
Let $\mu^*$ and $\mu^{\star}$ be the optimal solutions to the optimal problem $\min_{\mu\in\mathbb{R}}\{\tilde{V}_{\lambda}^1(\mu,q)-p^T\mu\}$ and $\min_{\mu\in\mathbb{R}}\{\tilde{V}_{\lambda,T}^1(\mu,q)-p^T\mu\}$, respectively. If $\min_{\mu\in\mathbb{R}}\{\tilde{V}_{\lambda}^1(\mu,q)-p^T\mu\} \geq \min_{\mu\in\mathbb{R}}\{\tilde{V}_{\lambda,T}^1(\mu,q)-p^T\mu\}$, then we have $|V_\lambda(p,q)-V_{\lambda,T}(p,q)|\leq |\tilde{V}_{\lambda}^1(\mu^\star,q)-\tilde{V}_{\lambda,T}^1(\mu^\star,q)|$. Otherwise, it is true that $|V_\lambda(p,q)-V_{\lambda,T}(p,q)|\leq |\tilde{V}_{\lambda}^1(\mu^*,q)-\tilde{V}_{\lambda,T}^1(\mu^*,q)|$. Therefore, we have, for any $p\in\Delta(K)$ and $q\in \Delta(L)$, $|V_\lambda(p,q)-V_{\lambda,T}(p,q)| \leq \|\tilde{V}_{\lambda}^1(\mu,q)-\tilde{V}_{\lambda,T}^1(\mu,q)\|_{\sup}$, which implies that $\|V_\lambda-V_{\lambda,T}\|_{\sup}\leq\|\tilde{V}_\lambda^1-\tilde{V}_{\lambda,T}^1\|_{\sup}$.

Next, we show that $\|V_\lambda-V_{\lambda,T}\|_{\sup}\geq\|\tilde{V}_\lambda^1-\tilde{V}_{\lambda,T}^1\|_{\sup}$. According to equation (\ref{eq: game value   relations 1, discounted game}) and (\ref{eq: game   value relation 1, truncated discounted game}), we have
\begin{align*}
 &|\tilde{V}_\lambda^1(\mu,q)-\tilde{V}_{\lambda,T}^1(\mu,q)|\\
 =&|\max_{p\in \Delta(K)}\{V_\lambda(p,q)+p^T\mu\}-\max_{p\in \Delta(K)}\{V_{\lambda,T}(p,q)+p^T\mu\}|.
\end{align*}
Let $p^*$ and $p^\star$ be the optimal solutions to the optimal problems $\max_{p\in \Delta(K)}\{V_\lambda(p,q)+p^T\mu\}$ and $\max_{p\in \Delta(K)}\{V_{\lambda,T}(p,q)+p^T\mu\}$, respectively. If $\max_{p\in \Delta(K)}\{V_\lambda(p,q)+p^T\mu\}\geq\max_{p\in \Delta(K)}\{V_{\lambda,T}(p,q)+p^T\mu\}$, then we have $|\tilde{V}_\lambda^1(\mu,q)-\tilde{V}_{\lambda,T}^1(\mu,q)| \leq |V_\lambda(p^*,q)-V_{\lambda,T}(p^*,q)|.$ Otherwise, it is true that $|\tilde{V}_\lambda^1(\mu,q)-\tilde{V}_{\lambda,T}^1(\mu,q)| \leq |V_\lambda(p^\star,q)-V_{\lambda,T}(p^\star,q)|.$ Therefore, we conclude that for any $\mu\in\mathbb{R}^{|K|}$ and $q\in \Delta(L)$, $|\tilde{V}_\lambda^1(\mu,q)-\tilde{V}_{\lambda,T}^1(\mu,q)| \leq \|V_\lambda(p,q)-V_{\lambda,T}(p,q)\|_{\sup}$, which implies that $\|\tilde{V}_\lambda^1-\tilde{V}_{\lambda,T}^1\|\leq \|V_\lambda(p,q)-V_{\lambda,T}(p,q)\|_{\sup}$.

Therefore, we prove the first equality of (\ref{eq: game value equalities}). Following the same steps, we have $\|V_\lambda-V_{\lambda,T}\|_{\sup}=\|\tilde{V}_\lambda^2-\tilde{V}_{\lambda,T}^2\|_{\sup}$.
\end{proof}

When we use $\tilde{V}_{\lambda,T}^1(\mu,q)$ and $\tilde{V}_{\lambda,T}^2(p,\nu)$ to approximate $\tilde{V}_\lambda^1(\mu,q)$ and $\tilde{V}_\lambda^2(p,\nu)$, respectively, we are interested in how fast the approximations converge to the real game values. To this purpose, we define two operators $\tilde{F}^1$ and $\tilde{F}^2$ as
\begin{align}
  \tilde{F}^{1,\tilde{V}^1}_z(\mu,q)=&\max_{a\in A}(1-\lambda)\sum_{b\in B} \bar{z}_{q,z}(b)\nonumber\\
  &\tilde{V}^1(\mu^+(\mu,a,b,z,q),q^+(b,z,q)), \label{eq: tilde F 1}\\
  \tilde{F}^{2,\tilde{V}^2}_r(p,\nu)=&\min_{b\in B}(1-\lambda)\sum_{a\in A} \bar{r}_{p,r}(a)\nonumber\\
  &\tilde{V}^2(p^+(a,r,p),\nu^+(\nu,a,b,r,p)). \label{eq: tilde F 2}
\end{align}
The two operators $\tilde{F}^1$ and $\tilde{F}^2$ are contraction mappings.
\begin{lemma}
  \label{lemma: contraction mappings}
Given any $z\in \Delta(K)$, $r\in \Delta(L)$ and $\lambda\in (0,1)$, the operators $\tilde{F}^1$ and $\tilde{F}^2$ defined in (\ref{eq: tilde F 1}) and (\ref{eq: tilde F 2}) are contraction mappings with contraction constant $1-\lambda$, i.e.
\begin{align}
  \|\tilde{F}^{1,\tilde{V}^1_1}_z-\tilde{F}^{1,\tilde{V}^1_2}_z\|_{\sup}\leq (1-\lambda)\|\tilde{V}_1^1-\tilde{V}_2^1\|_{\sup}, \label{eq: contraction F 1}\\
  \|\tilde{F}^{2,\tilde{V}^2_1}_r-\tilde{F}^{2,\tilde{V}^2_2}_r\|_{\sup}\leq (1-\lambda)\|\tilde{V}_1^2-\tilde{V}_2^2\|_{\sup}, \label{eq: contraction F 2}
\end{align}
where $\tilde{V}^1_{1,2}:\mathbb{R}^{|K|}\times\Delta(L)\rightarrow \mathbb{R}$ and $\tilde{V}^2_{1,2}:\Delta(K)\times \mathbb{R}^{|L|} \rightarrow \mathbb{R}$.
\end{lemma}
\begin{proof}
Let $a^*$ and $a^\star$ be the optimal solutions to the optimal problems $\max_{a\in A}(1-\lambda)\sum_{b\in B} \bar{z}_{q,z}\tilde{V}^1_1(\mu^+(\mu,a,b,z,q),q^+(b,z,q))$ and $\max_{a\in A}(1-\lambda)\sum_{b\in B} \bar{z}_{q,z}\tilde{V}^1_2(\mu^+(\mu,a,b,z,q),q^+(b,z,q))$.

If $\tilde{F}^{1,\tilde{V}^1_1}_z(\mu,q)\geq\tilde{F}^{1,\tilde{V}^1_2}_z(\mu,q)$, we have
\begin{align*}
  &|\tilde{F}^{1,\tilde{V}^1_1}_z(\mu,q)-\tilde{F}^{1,\tilde{V}^1_2}_z(\mu,q)|\\
  \leq&(1-\lambda)\sum_{b\in B} \bar{z}_{q,z}(b)|\tilde{V}^1_1(\mu^+(\mu,a^*,b,z,q),q^+(b,z,q))\\
  &-\tilde{V}^1_2(\mu^+(\mu,a^*,b,z,q),q^+(b,z,q))|\\
  \leq & (1-\lambda)\sum_{b\in B}\bar{z}_{q,z}(b) \|\tilde{V}^1_1-\tilde{V}^1_2\|_{\sup}\\
  =& (1-\lambda)\|\tilde{V}^1_1-\tilde{V}^1_2\|_{\sup}.
  \end{align*}

Otherwise, we have
\begin{align*}
  &|\tilde{F}^{1,\tilde{V}^1_1}_z(\mu,q)-\tilde{F}^{1,\tilde{V}^1_2}_z(\mu,q)|\\
 \leq&(1-\lambda)\sum_{b\in B} \bar{z}_{q,z}(b)|\tilde{V}^1_1(\mu^+(\mu,a^\star,b,z,q),q^+(b,z,q))\\
 &-\tilde{V}^1_2(\mu^+(\mu,a^\star,b,z,q),q^+(b,z,q))|\\
  \leq & (1-\lambda)\sum_{b\in B}\bar{z}_{q,z}(b) \|\tilde{V}^1_1-\tilde{V}^1_2\|_{\sup}= (1-\lambda)\|\tilde{V}^1_1-\tilde{V}^1_2\|_{\sup}.
\end{align*}

Hence, for any $\mu\in\mathbb{R}^{|K|}$ and $q\in \Delta(L)$, $|\tilde{F}^{1,\tilde{V}^1_1}_z(\mu,q)-\tilde{F}^{1,\tilde{V}^1_2}_z(\mu,q)|\leq (1-\lambda)\|\tilde{V}^1_1-\tilde{V}^1_2\|_{\sup}$, which implies equation (\ref{eq: contraction F 1}). Equation (\ref{eq: contraction F 2}) is shown following the same steps.
\end{proof}

Lemma \ref{lemma: contraction mappings} further implies that our game value approximations $\tilde{V}_{\lambda,T}^1(\mu,q)$ and $\tilde{V}^2_{\lambda,T}(p,\nu)$ converge to the real game values $\tilde{V}_\lambda^1(\mu,q)$ and $\tilde{V}^2_\lambda(p,\nu)$ exponentially fast over $T$.
\begin{theorem}
  \label{theorem: exponential fast convergence}
Consider the $\lambda$-discounted repeated Bayesian dual games $\tilde{\Gamma}_\lambda^1(\mu,q)$ and $\tilde{\Gamma}_\lambda^2(p,\nu)$, and their game values $\tilde{V}_\lambda^1(\mu,q)$ and $\tilde{V}_\lambda^2(p,\nu)$. The game values $\tilde{V}_{\lambda,T}^1(\mu,q)$ and $\tilde{V}_{\lambda,T}^2(p,\nu)$ of $\lambda$-discounted $T$-stage dual games $\tilde{\Gamma}_{\lambda,T}^1(\mu,q)$ and $\tilde{\Gamma}_{\lambda,T}^2(p,\nu)$ converge to $\tilde{V}_\lambda^1(\mu,q)$ and $\tilde{V}_\lambda^2(p,\nu)$ exponentially fast with respect to the time horizon $T$ with convergence rate $1-\lambda$, i.e.
\begin{align}
  \|\tilde{V}_\lambda^1-\tilde{V}_{\lambda,T+1}^1\|_{\sup}\leq & (1-\lambda)\|\tilde{V}_\lambda^1-\tilde{V}_{\lambda,T}^1\|_{\sup},\label{eq: convergence rate dual game 1}\\
  \|\tilde{V}_\lambda^2-\tilde{V}_{\lambda,T+1}^2\|_{\sup}\leq & (1-\lambda)\|\tilde{V}_\lambda^2-\tilde{V}_{\lambda,T}^2\|_{\sup}.\label{eq: convergence rate dual game 2}
\end{align}
\end{theorem}
\begin{proof}
Equation (\ref{eq: dynamic pragraming,   discounted dual game 1}) and (\ref{eq: dynamic programming,   truncated discounted dual game 1}) and the definition of $\tilde{F}^1$ in (\ref{eq: tilde F 1}) imply that
\begin{align*}
  &|\tilde{V}^1_\lambda(\mu,q)-\tilde{V}^1_{\lambda,T+1}(\mu,q)|\\
  =&|\min_{z\in \Delta(L)} \tilde{F}^{1,\tilde{V}^1_\lambda}_z(\mu,q)-\min_{z\in \Delta(L)}\tilde{F}^{1,\tilde{V}^1_{\lambda,T}}_z(\mu,q)|.
\end{align*}

Let $z^*$ and $z^\star$ be the optimal solutions to the optimal problems $\min_{z\in \Delta(L)} \tilde{F}^{1,\tilde{V}^1_\lambda}_z(\mu,q)$ and $\min_{z\in \Delta(L)}\tilde{F}^{1,\tilde{V}^1_{\lambda,T}}_z(\mu,q)$, respectively. If $\min_{z\in \Delta(L)} \tilde{F}^{1,\tilde{V}^1_\lambda}_z(\mu,q)\geq\min_{z\in \Delta(L)}\tilde{F}^{1,\tilde{V}^1_{\lambda,T}}_z(\mu,q)$, according to equation (\ref{eq: contraction F 1}), we have
\begin{align*}
  &|\tilde{V}^1_\lambda(\mu,q)-\tilde{V}^1_{\lambda,T+1}(\mu,q)| \leq |\tilde{F}^{1,\tilde{V}^1_\lambda}_{z^\star}(\mu,q)-\tilde{F}^{1,\tilde{V}^1_{\lambda,T}}_{z^\star}(\mu,q)|\\
  \leq & \|\tilde{F}^{1,\tilde{V}^1_\lambda}_{z^\star}-\tilde{F}^{1,\tilde{V}^1_{\lambda,T}}_{z^\star}\|_{\sup}\leq (1-\lambda)\|\tilde{V}^1_\lambda-\tilde{V}^1_{\lambda,T}\|_{\sup}.
\end{align*}
Otherwise, we have
\begin{align*}
  &|\tilde{V}^1_\lambda(\mu,q)-\tilde{V}^1_{\lambda,T+1}(\mu,q)| \leq |\tilde{F}^{1,\tilde{V}^1_\lambda}_{z^*}(\mu,q)-\tilde{F}^{1,\tilde{V}^1_{\lambda,T}}_{z^*}(\mu,q)|\\
  \leq & \|\tilde{F}^{1,\tilde{V}^1_\lambda}_{z^*}-\tilde{F}^{1,\tilde{V}^1_{\lambda,T}}_{z^*}\|_{\sup}\leq (1-\lambda)\|\tilde{V}^1_\lambda-\tilde{V}^1_{\lambda,T}\|_{\sup}.
\end{align*}

Hence, for any $\mu\in\mathbb{R}^{|K|}$ and $q\in\Delta(L)$, $|\tilde{V}^1_\lambda(\mu,q)-\tilde{V}^1_{\lambda,T+1}(\mu,q)|\leq (1-\lambda)\|\tilde{V}^1_\lambda-\tilde{V}^1_{\lambda,T}\|_{\sup}$, which implies (\ref{eq: convergence rate dual game   1}). Equation (\ref{eq: convergence rate dual game   2}) is shown following the same steps.
\end{proof}

With the approximated values $\tilde{V}_{\lambda,T}^1(\mu,q)$ and $\tilde{V}_{\lambda,T}^2(p,\nu)$, we can use them in equation (\ref{eq: dynamic pragraming,   discounted dual game 1}) and (\ref{eq: dynamic pragraming,   discounted dual game 2}), and derive player 1 and 2's approximated security strategies $\tilde{\sigma}^\dag$ and $\tilde{\tau}^\dag$ as
\begin{align}
  \tilde{\sigma}^\dag(:,p,\nu)=&\argmax_{r\in\Delta(A)^{|K|}} \min_{b\in B} (1-\lambda) \sum_{a\in A} \bar{r}_{p,r}(a)\nonumber\\ &\tilde{V}^2_{\lambda,T}(p^+(a,r,p),\nu^+(\nu,a,b,r,p)), \label{eq: approximated security strategy, player 1}
\end{align}
\begin{align}
  \tilde{\tau}^\dag(:,\mu,q)=& \argmin_{z\in\Delta(B)^{|L|}}\max_{a\in A}(1-\lambda)\sum_{b\in B}\bar{z}_{q,z}(b)\nonumber\\
   &\tilde{V}^1_{\lambda,T}(\mu^+(\mu,a,b,z,q),q^+(b,z,q)). \label{eq: approximated security strategy, player 2}
\end{align}

Comparing equation (\ref{eq: approximated security strategy, player 1}) and (\ref{eq: approximated security strategy, player 2}) with equation (\ref{eq: dynamic programming, truncated discounted dual game 2}) and (\ref{eq: dynamic programming, truncated discounted dual game 1}), we see that the approximated security strategy of player 1 in discounted type 2 dual game is the security strategy of player 1 at stage 1 in $T+1$-stage discounted type 2 dual game, and the approximated security strategy of player 2 in discounted type 1 dual game is the security strategy of player 2 at stage 1 in $T+1$-stage discounted type 1 dual game. Following the same steps as in the proof of Theorem \ref{theorem: LP dual game T stage}, we provide LP formulations to compute the approximated security strategies $\tilde{\sigma}^\dag$ and $\tilde{\tau}^\dag$ in the following theorem.
\begin{theorem}
  \label{theorem: LP dual game discounted}
Consider a $T+1$-stage $\lambda$-discounted dual game $\tilde{\Gamma}_{\lambda,T+1}^2(p,\nu)$. Its game value $\tilde{V}_{\lambda,T+1}^2(p,\nu)$ satisfies
\begin{align}
&\tilde{V}_{\lambda,T+1}^2(p,\nu)= \max_{x\in X,u_{\lambda,T+1}\in U,u_{:,0;\lambda,T+1}\in\mathbb{R}^{|L|},\tilde{u}\in \mathbb{R}}\tilde{u} \label{eq: LP discounted dual game 1} \\
s.t. & \nu+u_{:,0;\lambda,T+1}\geq \tilde{u}\mathbf{1} \label{eq: LP discounted dual game 1-0}\\
& \lambda\sum_{k\in K}{M^{kl}}^T x_{k,h_1^A,h_1^B}+(1-\lambda){u_{l,h_1^A,h_1^B;\lambda,T+1}}^T\mathbf{1}\nonumber\\
 &\geq u_{l,0;\lambda,T+1}\mathbf{1}, \forall l\in L,\label{eq: LP discounted dual game 1-1}\\
& \lambda\sum_{k\in K}{M^{kl}}^T x_{k,h_{t+1}^A,h_{t+1}^B}+(1-\lambda){u_{l,h_{t+1}^A,h_{t+1}^B;\lambda,T+1}}^T\mathbf{1} \nonumber\\
&\geq u_{l,h_t^A,h_t^B;\lambda,T+1}^{a_t,b_t}\mathbf{1}, \forall t=1,\ldots,T, l\in L,h_{t+1}\in H_{t+1}^A,\nonumber\\
&\forall h_{t+1}^B\in H_{t+1}^B,\label{eq: LP discounted dual game 1-2}
\end{align}
where $X$ is a set including all real vectors satisfying (\ref{eq: x constraint 1}-\ref{eq: x constraint 3}) with $x_{:,h_0^A,h_0^B}^{a_0}=p$, and $U$ is an appropriately dimensional real space. The approximated security strategy $\tilde{\sigma}^\dag(:,p,\nu)$ of player 1 in the discounted type 2 dual game $\tilde{\Gamma}^2_\lambda(p,\nu)$ is his security strategy at stage $1$ in the $T+1$-stage discounted type 2 dual game $\tilde{\Gamma}_{\lambda,T+1}^2(p,\nu)$, and is computed as below
\begin{align}
  \tilde{\sigma}^{\dag}(k,p,\nu)=\frac{x^\star_{k,h_1^A,h_1^B}}{p^k}, \forall k\in K. \label{eq: approximated strategy form, player 1}
\end{align}

Similarly, the game value $\tilde{V}_{\lambda,T+1}^1(\mu,q)$ of a $T+1$-stage discounted type 1 dual game $\tilde{\Gamma}_{\lambda,T+1}^1(\mu,q)$ satisfies
\begin{align}
&\tilde{V}_{\lambda,T+1}^1(\mu,q)= \min_{\substack{y\in Y,w_{\lambda,T+1}\in W,\\w_{:,0;\lambda,T+1}\in\mathbb{R}^{|K|},\tilde{w}\in\mathbb{R}}}\tilde{w} \label{eq: LP discounted dual game player 2}\\
s.t. & \mu+w_{:,0;\lambda,T+1}\leq \tilde{w}\mathbf{1} \label{eq: LP discounted dual game player 2-0}\\
&\lambda\sum_{l\in L} M^{kl} y_{l,h_1^A,h_1^B}+(1-\lambda)w_{k,h_1^A,h_1^B;\lambda,T+1}\mathbf{1}\nonumber\\
&\leq w_{k,0;\lambda,T+1}\mathbf{1},\forall k\in K, \label{eq: LP discounted dual game 2-1}\\
  & \lambda \sum_{l\in L} M^{kl} y_{l,h_{t+1}^A,h_{t+1}^B}+ (1-\lambda) w_{k,h_{t+1}^A,h_{t+1}^B;\lambda,T+1}\mathbf{1}\nonumber\\
   &\leq w_{k,h_{t}^A,h_{t}^B;\lambda,T+1}^{a_{t},b_{t}}\mathbf{1},\forall t=1,\ldots,T,k\in K,h_{t+1}^A\in H_{t+1}^A,\nonumber\\
  &\forall h_{t+1}^B\in H_{t+1}^B,\label{eq: LP discounted dual game 2-2}
\end{align}
where $Y$ is a set including all real vectors satisfying (\ref{eq: y constraint 1}-\ref{eq: y constraint 3}) with $y_{:,h_0^A,h_0^B}^{b_0}=q$, and $W$ is an appropriately dimensional real space. The approximated security strategy $\tilde{\tau}^\dag(:,\mu,q)$ of player 2 in the discounted type 1 dual game $\tilde{\Gamma}^1_\lambda(\mu,q)$ is his security strategy at stage $1$ in $T+1$-stage discounted type 1 dual game $\tilde{\Gamma}^1_{\lambda,T+1}(\mu,q)$, and is computed as below.
\begin{align}
  \tilde{\tau}^\dag(l,\mu,q)=\frac{y^\star_{l,h_1^A,h_1^B}}{q^l}, \forall l\in L. \label{eq: approximated strategy form, player 2}
\end{align}
\end{theorem}

Corollary \ref{theorem: security strategy and sufficient statitics in discounted games} says that player 1 and 2's security strategies in discounted primal game $\Gamma_\lambda(p,q)$ are their security strategies in dual game $\tilde{\Gamma}_\lambda^2(p,\nu^*)$ and $\tilde{\Gamma}_\lambda^1(\mu^*,q)$, respectively, where $\mu^*$ and $\nu^*$ are the solutions to the optimal problems on the right hand side of (\ref{eq: game value   relations 2, discounted game}) and (\ref{eq: game value   relations 4, discounted game}).
Now, we know how to construct LP formulations to approximate the initial regrets, and players' security strategies in the dual games. We give the algorithms to compute the approximated security strategies for both players as below.
\begin{algorithm}{Player 1's approximated security strategy in discounted game $\Gamma_\lambda(p,q)$ \hfill{}}
\label{algorithm: player 1 discounted game}
\begin{enumerate}
  \item Initialization
  \begin{itemize}
    \item Set $T$, and read parameters: $k$, $M$, $p$ and $q$.
    \item Given $(p,q)$, compute $u_{:,0;\lambda,T}^\star$ according to the LP (\ref{eq: LP discounted game 1}-\ref{eq: LP discounted game 1-2}).
    \item Set $t=1$, $p_1=p$ and $\nu_1=-u_{:,0;\lambda,T}^\star$.
  \end{itemize}
  \item Compute player 1's approximated security strategy $\tilde{\sigma}^\dag(:,p_t,\nu_t)$ according to (\ref{eq:   approximated strategy form, player 1}) based on the LP (\ref{eq: LP discounted dual game 1}-\ref{eq: LP discounted dual game 1-2}) with $p=p_t$ and $\nu=\nu_t$.
  \item Choose an action in $A$ according to $\tilde{\sigma}^\dag(k,p_t,\nu_t)$, and announce it publicly. Read player 2's action $b_t$.
  \item Update $p_{t+1}$ and $\nu_{t+1}$ according to (\ref{eq: belief state player 1}) and (\ref{eq: nu discounted game}), respectively.
  \item Update $t=t+1$ and go to step $2$.
\end{enumerate}
\end{algorithm}

\begin{algorithm}{Player 2's approximated security strategy in discounted game $\Gamma_\lambda(p,q)$ \hfill{}}
\label{algorithm: player 2 discounted game}
\begin{enumerate}
  \item Initialization
    \begin{itemize}
      \item Set $T$, and read parameters: $l$, $M$, $p$ and $q$.
      \item Given $(p,q)$, compute $w_{:,0;\lambda,T}^\star$ according to the LP (\ref{eq: LP discounted game player 2}-\ref{eq: LP discounted game 2-2}).
      \item Set $t=1$, $q_1=q$ and $\mu_1=-w_{:,0;\lambda,T}^\star$.
    \end{itemize}
  \item Compute player 2's approximated security strategy $\tilde{\tau}^\dag(:,\mu_t,q_t)$ according to (\ref{eq: approximated   strategy form, player 2}) based on the LP (\ref{eq: LP discounted dual game player 2}-\ref{eq: LP discounted dual game 2-2}) with $\mu=\mu_t$ and $q=q_t$.
  \item Choose an action in $B$ according to $\tilde{\tau}^\dag(l,\mu_t,q_t)$, and announce it to the public. Read player 1's action $a_t$.
  \item Update $q_{t+1}$ and $\mu_{t+1}$ according to (\ref{eq: belief state player 2}) and (\ref{eq: mu   discounted game}), respectively.
  \item Update $t=t+1$ and go to step $2$.
\end{enumerate}
\end{algorithm}

\subsection{Performance analysis of the approximated security strategies}
With player 1 and 2's approximated security strategies $\tilde{\sigma}^\dag$ and $\tilde{\tau}^\dag$ described in Algorithm \ref{algorithm: player 1 discounted game} and \ref{algorithm: player 2 discounted game}, we are interested in their worst case payoffs $J^{\tilde{\sigma}^\dag}$ and $J^{\tilde{\tau}^\dag}$. Given player 1's strategy $\sigma\in \Sigma$, its worst case payoff in discounted game $\Gamma_\lambda(p,q)$ is defined as
\begin{align}
  J^\sigma(p,q)=\min_{\tau\in \mathcal{T}} \gamma_\lambda(p,q,\sigma,\tau).
\end{align}
Similarly, given player 2's strategy $\tau\in \mathcal{T}$, its worst case payoff in discounted game $\Gamma_\lambda(p,q)$ is defined as
\begin{align}
  J^\tau(p,q)=\max_{\sigma\in \Sigma}\gamma_\lambda(p,q,\sigma,\tau).
\end{align}
Because players' approximated security strategies in game $\Gamma_\lambda(p,q)$ are derived from the approximated security strategies in its dual games, their worst case payoffs in $\Gamma_\lambda(p,q)$ are highly related to the worst case payoffs in the dual games. We define player 1's worst case payoff $\tilde{J}^{2,\sigma}$ in dual game $\tilde{\Gamma}^2_\lambda(p,\nu)$ and player 2's worst case payoff $\tilde{J}^{1,\tau}$ in dual game $\tilde{\Gamma}^1_\lambda(\mu,q)$ as
\begin{align}
  \tilde{J}^{2,\sigma}(p,\nu)=&\min_{q\in \Delta(L)}\min_{\tau\in \mathcal{T}}\tilde{\gamma}^2_\lambda(p,\nu,q,\sigma,\tau),\label{eq: tilde J 2}\\
  \tilde{J}^{1,\tau}(\mu,q)=&\max_{p\in \Delta(K)}\max_{\sigma\in \Sigma} \tilde{\gamma}^1_\lambda(\mu,q,p,\sigma,\tau). \label{eq: tilde J 1}
\end{align}
Following the same steps as in the proof of (\ref{eq: game value   relations 1, discounted game}-\ref{eq: game value   relations 4, discounted game}) in \cite{de1996repeated,sorin2002first}, we can show the relations between $J^\sigma(p,q)$ and $\tilde{J}^{2,\sigma}(p,\nu)$, and between $J^\tau(p,q)$ and $\tilde{J}^{1,\tau}(\mu,q)$ as below.
\begin{align}
  \tilde{J}^{2,\sigma}(p,\nu)=\min_{q\in \Delta(L)}\{J^\sigma(p,q)+q^T\nu\},\label{eq: worst case payoff relation 1}\\
  J^\sigma(p,q)=\max_{\nu\in\mathbb{R}^{|L|}}\{\tilde{J}^{2,\sigma}(p,\nu)-q^T\nu\}, \label{eq: worst case payoff relation 2}\\
  \tilde{J}^{1,\tau}(\mu,q)=\max_{p\in \Delta(K)}\{ J^\tau(p,q)+p^T\mu\},\label{eq: worst case payoff relation 3}\\
  J^\tau(p,q)=\min_{\mu\in\mathbb{R}^{|K|}}\{\tilde{J}^{1,\tau}(\mu,q)-p^T\mu\}. \label{eq: worst case payoff relation 4}
\end{align}

The worst case payoffs $\tilde{J}^{2,\sigma}$ and $\tilde{J}^{1,\tau}$ satisfy recursive formulas if $\sigma$ and $\tau$ are stationary strategies, i.e. $\sigma$ only depends on $p_t$ and $\nu_t$, and $\tau$ only depends on $\mu_t$ and $q_t$.
\begin{lemma}
  \label{lemma: worst case payoff recursive formula}
Let $\sigma$ be player 1's stationary strategy that depends only on $p_t$ and $\nu_t$ in the discounted type 2 dual game $\tilde{\Gamma}_\lambda^2(p,\nu)$. Its worst case payoff $\tilde{J}^{2,\sigma}(p,\nu)$ satisfies \begin{align}
\tilde{J}^{2,\sigma}(p,\nu)=\tilde{F}^{2,\tilde{J}^{2,\sigma}}_{\sigma(:,p,\nu)}(p,\nu). \label{eq: recursive tilde J 2}
\end{align}

Similarly, let $\tau$ be player 2's stationary strategy that depends only on $\mu_t$ and $q_t$ in the discounted type 1 dual game $\tilde{\Gamma}_\lambda^1(\mu,q)$. Its worst case payoff $\tilde{J}^{1,\tau}(\mu,q)$ satisfies \begin{align}
  \tilde{J}^{1,\tau}(\mu,q)=\tilde{F}^{1,\tilde{J}^{1,\tau}}_{\tau(:,\mu,q)}(\mu,q).    \label{eq: recursive tilde J1}
\end{align}
\end{lemma}
\begin{proof}
According to Bellman's principle, we have
\begin{align*}
&J^\sigma(p,q)=\min_{z\in \Delta(B)^{|L|}}\left(\lambda\sum_{l\in L,k\in K}p^kq^l{r^k}^TM^{kl}z^l+(1-\lambda)\right.\\
&\left.\sum_{a\in A,b\in B}\bar{r}_{p,r}(a)\bar{z}_{q,z}(b)J^\sigma(p^+(a,p,r),q^+(b,q,z))\right),
\end{align*}
where $r^k=\sigma(k,p,\nu)$.
From equation (\ref{eq: worst case payoff   relation 1}), we derive that
\begin{align*}
&\tilde{J}^{2,\sigma}(p,\nu)=\min_{q\in \Delta(L),z\in \Delta(B)^{|L|}}\left(\sum_{b\in B,l\in L}q^lz^l(b)\nu^l\right.\\
\end{align*}
\begin{align*}
&\left.+\lambda\sum_{l\in L,k\in K,a\in A,b\in B}p^kq^lr^k(a)M^{kl}_{ab}z^l(b)+(1-\lambda)\right.\\
&\left.\sum_{a\in A,b\in B}\bar{r}_{p,r}(a)\bar{z}_{q,z}(b)J^\sigma(p^+(a,p,r),q^+(b,q,z)) \right)\\
&=(1-\lambda)\min_{q\in\Delta(L),z\in \Delta(B)^{|L|}}\sum_{b\in B}\bar{z}_{q,z}(b)\sum_{a\in A}\bar{r}_{p,r}(a)\\
&\left(\sum_{l\in L}q^{+l}(b,q,z)\frac{\nu^l+\lambda\sum_{k\in K}p^{+k}(a,p,r)M^{kl}_{ab}}{1-\lambda}\right.\\
&\left.+J^\sigma(p^+(a,p,r),q^+(b,q,z))\right).
\end{align*}
Since $q^lz^l(b)=q^{+l}(b,q,z)\bar{z}_{q,z}(b)$ for any $l\in L$ and $b\in B$, the minimum function taken with respect to $q\in\Delta(L),z\in \Delta(B)^{|L|}$ is the same as the minimum function taken with respect to $q^+\in \Delta(L)^{|B|},\bar{z}\in\Delta(B)$. Thus, we have
\begin{align*}
  &\tilde{J}^{2,\sigma}(p,\nu)\\
  =&(1-\lambda)\min_{q^+\in \Delta(L)^{|B|},\bar{z}\in\Delta(B)}\sum_{b\in B}\bar{z}(b)\sum_{a\in A}\bar{r}_{p,r}(a)\left(\sum_{l\in L}q^{+l}(b) \right.\\
&\left. \frac{\nu^l+\lambda\sum_{k\in K}p^{+k}(a,p,r)M^{kl}_{ab}}{1-\lambda}+J^\sigma(p^+(a,p,r),q^+(b))\right)\\
=&(1-\lambda)\min_{\bar{z}\in\Delta(B)}\sum_{b\in B}\bar{z}(b)\sum_{a\in A}\bar{r}_{p,r}(a)\\
&\tilde{J}^{2,\sigma}(p^+(a,p,r),\nu^+(\nu,a,b,r,p))\\
=&(1-\lambda)\min_{b\in B}\sum_{a\in A}\bar{r}_{p,r}(a)\tilde{J}^{2,\sigma}(p^+(a,p,r),\nu^+(\nu,a,b,r,p))\\
=&\tilde{F}^{2,\tilde{J}^{2,\sigma}}_{\sigma(:,p,\nu)}(p,\nu).
\end{align*}
The second equality is derived from (\ref{eq: worst case   payoff relation 1}).

Following the same steps, we can show that $\tilde{J}^{1,\tau}(\mu,q)=\tilde{F}^{1,\tilde{J}^{1,\tau}}_{\tau(:,\mu,q)}(\mu,q)$.
\end{proof}

Based on Lemma \ref{lemma: worst case payoff recursive formula}, we are ready to analyze the performance difference between players' approximated security strategies and their security strategies.
\begin{theorem}
Consider a discounted game $\Gamma_\lambda(p,q)$. If player 2 uses $\tilde{\sigma}^\dag$ defined in (\ref{eq: approximated security   strategy, player 1}) as his strategy, and follows Algorithm \ref{algorithm: player 1 discounted game} to take actions, then his worst case payoff $J^{\tilde{\sigma}^\dag}(p,q)$ satisfies
\begin{align}
  \|J^{\tilde{\sigma}^\dag}-V_\lambda\|_{\sup}\leq \frac{2(1-\lambda)}{\lambda}\|V_{\lambda|T}-V_\lambda\|_{\sup}.\label{eq: performance difference 1}
\end{align}

If player 2 uses $\tilde{\tau}^\dag$ defined in (\ref{eq: approximated   security strategy, player 2}) as his strategy, and follows Algorithm \ref{algorithm: player 2 discounted game} to take actions, then his worst case payoff $J^{\tilde{\tau}^\dag}(p,q)$ satisfies
\begin{align}
  \|J^{\tilde{\tau}^\dag}-V_\lambda\|_{\sup}\leq \frac{2(1-\lambda)}{\lambda}\|V_{\lambda|T}-V_\lambda\|_{\sup}. \label{eq: performance difference 2}
\end{align}
\end{theorem}
\begin{proof}
According to equation (\ref{eq: worst case   payoff relation 2}) and (\ref{eq: game value   relations 2, discounted game}), we have $|J^{\tilde{\sigma}^\dag}(p,q)-V_\lambda(p,q)|=|\max_{\nu\in\mathbb{R}^{|L|}}\{\tilde{J}^{2,\tilde{\sigma}^\dag}(p,\nu)-q^T\nu\} -\max_{\nu\in\mathbb{R}^{|L|}}\{\tilde{V}_\lambda^2(p,\nu)-q^T\nu\}|$. Let $\nu^*$ be the solution to the optimal problem $\max_{\nu\in\mathbb{R}^{|L|}}\{\tilde{V}_\lambda^2(p,\nu)-q^T\nu\}$. Since $J^{\tilde{\sigma}^\dag}(p,q)\leq V_\lambda(p,q)$, we have
\begin{align}
  &|J^{\tilde{\sigma}^\dag}(p,q)-V_\lambda(p,q)| \leq |\tilde{J}^{2,\tilde{\sigma}^\dag}(p,\nu^*)-\tilde{V}_\lambda^2(p,\nu^*)|\nonumber\\
   \leq& \|\tilde{J}^{2,\tilde{\sigma}^\dag}-\tilde{V}_\lambda^2\|_{\sup}, \forall p\in\Delta(K),q\in\Delta(L) \label{eq: temp result 1}
 \end{align}

According to equation (\ref{eq: recursive tilde J 2}), (\ref{eq: convergence rate dual game   2}) and (\ref{eq: contraction F 2}), we have for any $p\in\Delta(K)$ and any $\nu\in\mathbb{R}^{|L|}$,
\begin{align*}
&|\tilde{J}^{2,\tilde{\sigma}^\dag}(p,\nu)-\tilde{V}_\lambda^2(p,\nu)|\\
\leq & |\tilde{J}^{2,\tilde{\sigma}^\dag}(p,\nu)-\tilde{V}^2_{\lambda,T+1}(p,\nu)|+|\tilde{V}^2_{\lambda,T+1}(p,\nu)-\tilde{V}_\lambda^2(p,\nu)|\\
\leq & |\tilde{F}^{2,\tilde{J}^{2,\tilde{\sigma}^\dag}}_{\tilde{\sigma}^\dag(:,p,\nu)}(p,\nu)-\tilde{F}^{2,\tilde{V}^2_{\lambda,T}}_{\tilde{\sigma}^\dag(:,p,\nu)}(p,\nu)|+(1-\lambda)|\tilde{V}^2_{\lambda,T}-\tilde{V}^2_\lambda\|_{\sup}\\
\leq & (1-\lambda)\|\tilde{J}^{2,\tilde{\sigma}^\dag}-\tilde{V}^2_{\lambda,T}\|_{\sup}+(1-\lambda)\|\tilde{V}^2_{\lambda,T}-\tilde{V}^2_\lambda\|_{\sup}.
\end{align*}
Thus, we have $\|\tilde{J}^{2,\tilde{\sigma}^\dag}-\tilde{V}_\lambda^2\|_{\sup} \leq (1-\lambda)\|\tilde{J}^{2,\tilde{\sigma}^\dag}-\tilde{V}^2_{\lambda,T}\|_{\sup}+(1-\lambda)|\tilde{V}^2_{\lambda,T}-\tilde{V}^2_\lambda\|_{\sup}\leq (1-\lambda)\|\tilde{J}^{2,\tilde{\sigma}^\dag}-\tilde{V}^2_{\lambda}\|_{\sup}+2(1-\lambda)|\tilde{V}^2_{\lambda,T}-\tilde{V}^2_\lambda\|_{\sup}$, which implies that
\begin{align*}
  \|\tilde{J}^{2,\tilde{\sigma}^\dag}-\tilde{V}_\lambda^2\|_{\sup} \leq \frac{2(1-\lambda)}{\lambda}\|\tilde{V}^2_{\lambda,T}-\tilde{V}^2_\lambda\|_{\sup}.
\end{align*}

After applying the above inequality to (\ref{eq: temp result 1}), we have for any $p\in \Delta(K)$ and $q\in \Delta(L)$, $|J^{\tilde{\sigma}^\dag}(p,q)-V_\lambda(p,q)|\leq \frac{2(1-\lambda)}{\lambda}\|\tilde{V}^2_{\lambda,T}-\tilde{V}^2_\lambda\|_{\sup}$, which implies that $\|J^{\tilde{\sigma}^\dag}-V_\lambda\|_{\sup}\leq \frac{2(1-\lambda)}{\lambda}\|\tilde{V}^2_{\lambda,T}-\tilde{V}^2_\lambda\|_{\sup}$. According to equation (\ref{eq: game value equalities}), equation (\ref{eq: performance difference 1}) is shown.

With the same technique, equation (\ref{eq: performance difference 2}) can be shown to be true.
\end{proof}

\section{Case Study: Jamming in Underwater Sensor Networks}
The jamming in underwater sensor networks is originally modelled as a two-player zero-sum one-shot Bayesian game in \cite{vadori2015jamming}. We adopt the game model in \cite{vadori2015jamming}, and extend it to a repeated Bayesian game with uncertainties on both the sensors' positions and the jammer's position.

Let us assume that there are two sensors in the network which send data to a sink node through a shared spectrum at $[10,40]$ kHz. The distance from a sensor to the sink node is either $1$ km or $5$ km. The shared spectrum is divided into two channels, $\mathcal{B}_1=[10,25]$ kHz and $\mathcal{B}_2=[25,40]$ kHz. Generally speaking, channel 1 works much better for a sensor far away, and almost the same as channel 2 for a sensor close by. The sensors need to coordinate with each other to use the two channels to transfer as much data as possible to the sink node in the presence of a jammer. The jammer's distance from the sink node is $0.5$ km or $2$ km. While the jammer doesn't know the sensors' positions, the sensors don't know the jammer's position either. For every time period, the jammer can only generate noises in one channel, which can be detected by the sensors. At the same time, the jammer can also observe whether a channel is used by a far-away sensor or a close-by sensor. The jammer's goal is to minimize the data transmitted through the two channels.

The sensors (player 1) have three types according to their position distribution, which are $[1\ 1]$ (type 1), $[1\ 5]$ (type 2), and $[5,5]$ (type 3). We consider $[1\ 5]$ and $[5\ 1]$ as one type. The initial distribution over the three types is $p_0=[0.5\ 0.3\ 0.2]$. When playing the game, they have two choices, sensor 1 uses channel 1 while sensor 2 uses channel 2 (action 1) or sensor 1 uses channel 2 while sensor 2 uses channel 1 (action 2). The jammer (player 2) has two types according to his position, which are $0.5$ (type 1) and $2$ (type 2), and the initial distribution over the two types is $q_0=[0.5\ 0.5]$. His actions are jamming channel 1 (action 1) or channel 2 (action 2). Suppose both the sensors and the jammer transmit with constant power $95$ dB re $\nu$Pa. A channel's capacity can be computed based on the Shannon-Hartley theorem with the average under water signal-to-noise ratio described in \cite{baldo2008cognitive,vadori2015jamming}. The payoff matrices, whose element is the total channel capacity measured by bit/s given both players' types and actions, are given in Table \ref{Table: payoff matrice}.
\begin{table}
\caption{Total channel capacity}
\label{Table: payoff matrice}
\center
\begin{tabular}{|c|cc|cc|}
  \hline
  \backslashbox{k}{l}  & \multicolumn{2}{|c}{1(0.5 km)} & \multicolumn{2}{|c|}{2 (2 km)} \\ \hline
  \multirow{2}{*}{1 ([1 1] km)} & 108.89 & 113.78 & 122.30 & 154.40 \\
   & 108.89 & 113.78 & 122.30 & 154.40 \\ \hline
  \multirow{2}{*}{2 ([1 5] km)} & 11.48 & 107.38 & 24.89 & 107.42\\
   & 99.04 & 20.15 & 100.26 & 60.77 \\ \hline
  \multirow{2}{*}{3 ([5 5] km)} & 1.64 & 13.75 & 2.85 &13.79 \\
   & 1.64 & 13.75 & 2.85 &13.79 \\
  \hline
\end{tabular}
\end{table}

We first consider a two-stage Bayesian repeated game between the sensors and the jammer. Based on the linear program (\ref{eq: LP   player 1}-\ref{eq: LP plyaer 1-2}), we compute the sensors' security strategy shown in Table \ref{table: sensors' security strategy in sequence form} with a security level to be $162.49$ bit/s. According to the linear program (\ref{eq: LP   player 2}-\ref{eq: LP player 2 2}), the jammer's security strategy is computed, and given in Table \ref{table: jammer's security strategy}. The jammer's security level is $162.49$ bit/s which meets the sensors' security level. We then use the players' security strategies in Table \ref{table: sensors' security strategy in sequence form} and \ref{table: sensors' security strategy in sequence form} in the two-stage under water jamming game. The jamming game was run for $100$ times for each experiment, and we did the experiment for $30$ times. The total channel capacity in the jamming game varies from $142.12$ bit/s to $185.23$ bit/s with an average capacity to be $162.79$ bit/s, which is very close to the game value computed according to (\ref{eq: LP   player 1}-\ref{eq: LP plyaer 1-2}) and (\ref{eq: LP   player 2}-\ref{eq: LP player 2 2}).

Next, we would like to use security strategies based on fixed-sized sufficient statistics in the jamming game, and see whether we can still achieve the game value. First of all, we need to verify Theorem \ref{theorem: security strategy relation between primal and dual games}. According to Lemma \ref{lemma: optimal solution T stage} and linear program (\ref{eq: LP   player 2}-\ref{eq: LP player 2 2}) and (\ref{eq: LP   player 1}-\ref{eq: LP plyaer 1-2}), the initial regret $\mu^*$ in type 2 dual game $\tilde{\Gamma}^2_T(p_0,\mu^*)$ is $[-145.45\ -179.53]$, and the initial regret $\nu^*$ in type 1 dual game $\tilde{\Gamma}^1_T(\nu^*,q_0)$ is$[-234.77\ -141.44\ -13.38]$. Player 1's security strategy in dual game $\tilde{\Gamma}^2_T(p_0,\mu^*)$ is computed according to the linear program (\ref{eq: LP player 1 dual game}-\ref{eq: LP player 1-1 dual game}), and given in Table \ref{table: sensor's security strategy in dual game}. We see that player 1's security strategy in dual game $\tilde{\Gamma}^2_T(p_0,\mu^*)$ is different from but very close to player 1's security strategy in the primal game $\Gamma_T(p_0,q_0)$. The security level of $\tilde{\sigma}^*$ in the primal game is $162.49$ (checked by building a linear program the same as (\ref{eq: LP   player 1}-\ref{eq: LP plyaer 1-2}) with $x$ fixed), the game value of the primal game. Therefore, $\tilde{\sigma}^*$ is player 1's another security strategy. Player 2's security strategy in the dual game $\tilde{\Gamma}^1_T(\nu^*,q_0)$ is computed according to linear program (\ref{eq: LP player 2 dual game}-\ref{eq: LP player 2-2 dual game}), and given in Table \ref{table: jammer's security strategy in dual game}, which matches player 2's security strategy in the primal game $\Gamma_T(p_0,q_0)$. We then run the two-stage under water jamming game using security strategies based on fixed sized sufficient statistics, and followed Algorithm \ref{algorithm: player 1's strategy T stage} and \ref{algorithm: player 2's algorithm T stage} to take actions. For each experiment, the two-stage under water jamming game was run for $100$ times, and we did $30$ experiments. The channel capacity varies from $144.38$ to $180.31$ bit/s, with an average capacity to be $162.32$ bit/s, which is almost the same as the game value $162.49$ bit/s.
\begin{table}
  \caption{$\sigma^{1*}_t(k,h_t^A,h_t^B)$ in $\Gamma_T(p_0,q_0)$}
  \label{table: sensors' security strategy in sequence form}
  \center
\begin{tabular}{|c|c|c|c|c|c|}
  \hline
  \backslashbox{k}{$h_t^A,h_t^B$}& $\emptyset,\emptyset$ & 1,1 & 1,2 & 2,1 & 2,2 \\ \hline
  1 & 0.43& 0.5 &0.5& 0.5& 0.5 \\ \hline
  2& 0.18 & 0 &0.23 &0 &0.39 \\ \hline
  3& 0.44 & 0.5 & 0.5 & 0.5& 0.5 \\ \hline
\end{tabular}
\caption{$\tau^{1*}_t(l,h_t^A,h_t^B)$ in $\Gamma_T(p_0,q_0)$}
\label{table: jammer's security strategy}
\begin{tabular}{|c|c|c|c|c|c|}
  \hline
  \backslashbox{l}{$h_t^A,h_t^B$}& $\emptyset,\emptyset$ & 1,1 & 1,2 & 2,1 & 2,2 \\ \hline
  1 & 0.068& 0 & 0.5 & 0 &0.5 \\ \hline
  2& 1& 1& $\backslash$ & 1 & $\backslash$\\ \hline
\end{tabular}
\end{table}

\begin{table}
\caption{$\tilde{\sigma}^{1*}_t(k,h_t^A,h_t^B)$ in $\tilde{\Gamma}^2_T(p_0,\mu^*)$}
\label{table: sensor's security strategy in dual game}
\center
\begin{tabular}{|c|c|c|c|c|c|}
  \hline
  \backslashbox{k}{$h_t^A,h_t^B$}& $\emptyset,\emptyset$ & 1,1 & 1,2 & 2,1 & 2,2 \\ \hline
  1 & 0.33& 0.5 &0.5& 0.5& 0.5 \\ \hline
  2& 0.18 & 0 &0.25 &0.068 &0.38 \\ \hline
  3& 0.45 & 0.5 & 0.5 & 0.5& 0.5 \\ \hline
\end{tabular}
\caption{$\tilde{\tau}^{1*}_t(l,h_t^A,h_t^B)$ in $\tilde{\Gamma}^1_T(\nu^*,q_0)$}
\label{table: jammer's security strategy in dual game}
\begin{tabular}{|c|c|c|c|c|c|}
  \hline
  \backslashbox{l}{$h_t^A,h_t^B$}& $\emptyset,\emptyset$ & 1,1 & 1,2 & 2,1 & 2,2 \\ \hline
  1 & 0.068& 0 & 0.5 & 0 &0.5 \\ \hline
  2& 1& 1& $\backslash$ & 1 & $\backslash$\\ \hline
\end{tabular}
\end{table}

Finally, we test Algorithm \ref{algorithm: player 1 discounted game} and \ref{algorithm: player 2 discounted game} in the discounted under water jamming game with discount constant $\lambda=0.7$ to see whether the outcome satisfies our anticipation. In the algorithms, we set $T=3$, and $V_{\lambda,3}=78.28$ bit/s. First, we found that the highest game value of a $3$-stage discounted game occurs at $p_0=[1\ 0\ 0]$ and $q_0=[0\ 1]$, and $\|V_{\lambda,3}\|_{\sup}=118.99$ bit/s. Second, we found an upper bound on $\|V_\lambda(p_0,q_0)\|_{\sup}$. According to equation (\ref{eq: game value equalities}) and (\ref{eq: convergence rate dual   game 1}), we have $\|V_\lambda-V_{\lambda,3}\|_{\sup}\leq (1-\lambda)^3\|V_\lambda\|_{\sup}$, which implies that $\|V_\lambda\|_{\sup}\leq 1/(1-(1-\lambda)^3) \|V_{\lambda,3}\|_{\sup}=122.29$ bit/s. Third, we derive a lower bound on the security level of the sensors' approximated security strategy. According to equation (\ref{eq: performance difference   1}), (\ref{eq: game value equalities}) and (\ref{eq: convergence rate dual   game 2}), we have $J^{\tilde{\sigma}^\dag}(p_0,q_0)\geq V_{\lambda,3}(p_0,q_0)-2(1-\lambda)^4/\lambda\|V_\lambda\|_{\sup}\geq 75.44$ bit/s. Finally, we get an upper bound on the security level of the jammer's approximated security strategy. According to equation (\ref{eq: performance difference   2}) , (\ref{eq: game value equalities}) and (\ref{eq: convergence rate dual   game 1}), we have $J^{\tilde{\tau}^\dag}(p_0,q_0)\leq V_\lambda(p_0,q_0)+2(1-\lambda)^4/\lambda \|V_\lambda\|_{\sup}\leq V_{\lambda,3}(p_0,q_0)+(1-\lambda)^3\sum_{t=1}^\infty \lambda(1-\lambda)^{t-1} 154.4 +2(1-\lambda)^4/\lambda \|V_\lambda\|_{\sup}\leq 85.28$ bit/s. Therefore, our anticipated channel capacity in the discounted under water jamming game is between $75.44$ and $85.28$ bit/s. Now, we run the discounted under water jamming game (10 stages) for $100$ times. For each run, we truncate the infinite horizon discounted game to $10$ stages, since the total channel capacity for the truncated stages is less than $10^{-3}$ bit/s. The average channel capacity is $82.15$ bit/s, which is within our anticipation, and verifies our main results in the discounted games.

\section{Conclusion and future work}
This paper studies two-player zero-sum repeated Bayesian games, and provides LP formulations to compute players' security strategies in finite horizon case and approximated security strategies in discounted infinite horizon case with performance guarantee. In both cases, strategies based on fixed-sized sufficient statistic are provided. The fixed-sized sufficient statistics for each player consists of the belief over his own type and the regret with respect to the other player's type. We are interested in extending the results to two player zero-sum stochastic Bayesian games in the future. A main difference between a repeated Bayesian game and a stochastic Bayesian game is that there may not exist a Nash Equilibrium in a stochastic Bayesian game if the transition matrix depends on players' actions \cite{rosenberg2004stochastic}. Because of the difference, some results in repeated Bayesian games may not holds in stochastic Bayesian games, and hence further study is necessary.

\bibliography{}

\begin{thebibliography}{10}
\providecommand{\url}[1]{#1}
\csname url@samestyle\endcsname
\providecommand{\newblock}{\relax}
\providecommand{\bibinfo}[2]{#2}
\providecommand{\BIBentrySTDinterwordspacing}{\spaceskip=0pt\relax}
\providecommand{\BIBentryALTinterwordstretchfactor}{4}
\providecommand{\BIBentryALTinterwordspacing}{\spaceskip=\fontdimen2\font plus
\BIBentryALTinterwordstretchfactor\fontdimen3\font minus
  \fontdimen4\font\relax}
\providecommand{\BIBforeignlanguage}[2]{{%
\expandafter\ifx\csname l@#1\endcsname\relax
\typeout{** WARNING: IEEEtran.bst: No hyphenation pattern has been}%
\typeout{** loaded for the language `#1'. Using the pattern for}%
\typeout{** the default language instead.}%
\else
\language=\csname l@#1\endcsname
\fi
#2}}
\providecommand{\BIBdecl}{\relax}
\BIBdecl

\bibitem{rosenberg1998duality}
D.~Rosenberg, ``Duality and markovian strategies,'' \emph{International Journal
  of Game Theory}, vol.~27, no.~4, pp. 577--597, 1998.

\bibitem{nayyar2014common}
A.~Nayyar, A.~Gupta, C.~Langbort, and T.~Ba{\c{s}}ar, ``Common information
  based markov perfect equilibria for stochastic games with asymmetric
  information: Finite games,'' \emph{IEEE Transactions on Automatic Control},
  vol.~59, no.~3, pp. 555--570, 2014.

\bibitem{ouyang2017dynamic}
Y.~Ouyang, H.~Tavafoghi, and D.~Teneketzis, ``Dynamic games with asymmetric
  information: Common information based perfect bayesian equilibria and
  sequential decomposition,'' \emph{IEEE Transactions on Automatic Control},
  vol.~62, no.~1, pp. 222--237, 2017.

\bibitem{vasal2016systematic}
D.~Vasal and A.~Anastasopoulos, ``A systematic process for evaluating
  structured perfect bayesian equilibria in dynamic games with asymmetric
  information,'' in \emph{American Control Conference (ACC), 2016}.\hskip 1em
  plus 0.5em minus 0.4em\relax IEEE, 2016, pp. 3378--3385.

\bibitem{fudenberg1991perfect}
D.~Fudenberg and J.~Tirole, ``Perfect bayesian equilibrium and sequential
  equilibrium,'' \emph{journal of Economic Theory}, vol.~53, no.~2, pp.
  236--260, 1991.

\bibitem{aumann1995repeated}
R.~J. Aumann and M.~Maschler, \emph{Repeated games with incomplete
  information}.\hskip 1em plus 0.5em minus 0.4em\relax MIT press, 1995.

\bibitem{zamir1992repeated}
S.~Zamir, ``Repeated games of incomplete information: Zero-sum,''
  \emph{Handbook of Game Theory}, vol.~1, pp. 109--154, 1992.

\bibitem{bacsar1998dynamic}
T.~Ba{\c{s}}ar and G.~J. Olsder, \emph{Dynamic noncooperative game
  theory}.\hskip 1em plus 0.5em minus 0.4em\relax SIAM, 1998.

\bibitem{sorin2002first}
S.~Sorin, \emph{A first course on zero-sum repeated games}.\hskip 1em plus
  0.5em minus 0.4em\relax Springer Science \& Business Media, 2002, vol.~37.

\bibitem{von1996efficient}
B.~Von~Stengel, ``Efficient computation of behavior strategies,'' \emph{Games
  and Economic Behavior}, vol.~14, no.~2, pp. 220--246, 1996.

\bibitem{de1996repeated}
B.~De~Meyer, ``Repeated games and partial differential equations,''
  \emph{Mathematics of Operations Research}, vol.~21, no.~1, pp. 209--236,
  1996.

\bibitem{sion1958general}
M.~Sion \emph{et~al.}, ``On general minimax theorems,'' \emph{Pacific J. Math},
  vol.~8, no.~1, pp. 171--176, 1958.

\bibitem{sandholm2010state}
T.~Sandholm, ``The state of solving large incomplete-information games, and
  application to poker,'' \emph{AI Magazine}, vol.~31, no.~4, pp. 13--32, 2010.

\bibitem{vadori2015jamming}
V.~Vadori, M.~Scalabrin, A.~V. Guglielmi, and L.~Badia, ``Jamming in underwater
  sensor networks as a bayesian zero-sum game with position uncertainty,'' in
  \emph{2015 IEEE Global Communications Conference (GLOBECOM)}.\hskip 1em plus
  0.5em minus 0.4em\relax IEEE, 2015, pp. 1--6.

\bibitem{baldo2008cognitive}
N.~Baldo, P.~Casari, and M.~Zorzi, ``Cognitive spectrum access for underwater
  acoustic communications,'' in \emph{ICC Workshops-2008 IEEE International
  Conference on Communications Workshops}.\hskip 1em plus 0.5em minus
  0.4em\relax IEEE, 2008, pp. 518--523.

\bibitem{rosenberg2004stochastic}
D.~Rosenberg, E.~Solan, and N.~Vieille, ``Stochastic games with a single
  controller and incomplete information,'' \emph{SIAM journal on control and
  optimization}, vol.~43, no.~1, pp. 86--110, 2004.

\end{thebibliography}

\end{document}